\documentclass[11pt,letterpaper]{article}

\def\showauthornotes{0}
\def\showtableofcontents{1}
\def\showkeys{0}
\def\showdraftbox{0}
\def\showcolorlinks{1}
\def\usemicrotype{1}
\def\showfixme{0}

\usepackage{etex}

\usepackage[l2tabu, orthodox]{nag}

\usepackage{xspace,enumerate}

\usepackage[dvipsnames]{xcolor}

\usepackage[T1]{fontenc}
\usepackage[full]{textcomp}

\usepackage[american]{babel}

\usepackage{mathtools}

\usepackage{amsthm}

\newtheorem{theorem}{Theorem}[section]
\newtheorem*{theorem*}{Theorem}

\newtheorem{proposition}[theorem]{Proposition}
\newtheorem*{proposition*}{Proposition}
\newtheorem{lemma}[theorem]{Lemma}
\newtheorem*{lemma*}{Lemma}
\newtheorem{corollary}[theorem]{Corollary}
\newtheorem*{corollary*}{Corollary}
\newtheorem*{conjecture*}{Conjecture}
\newtheorem{fact}[theorem]{Fact}
\newtheorem*{fact*}{Fact}

\newtheorem*{hypothesis*}{Hypothesis}
\newtheorem{conjecture}[theorem]{Conjecture}

\theoremstyle{definition}
\newtheorem{definition}[theorem]{Definition}
\newtheorem{program}[theorem]{Program}

\newtheorem{problem}[theorem]{Problem}

\theoremstyle{remark}
\newtheorem{claim}[theorem]{Claim}
\newtheorem*{claim*}{Claim}
\newtheorem{remark}[theorem]{Remark}
\newtheorem*{remark*}{Remark}

\newtheorem*{observation*}{Observation}

\usepackage[
letterpaper,
top=1in,
bottom=1in,
left=1in,
right=1in]{geometry}

\usepackage{newpxtext} 
\usepackage{textcomp} 
\usepackage[varg,bigdelims]{newpxmath}
\usepackage[scr=rsfso]{mathalfa}
\usepackage{bm} 
\let\mathbb\varmathbb

\ifnum\showkeys=1
\usepackage[color]{showkeys}
\fi

\ifnum\showcolorlinks=1
\usepackage[
pagebackref,
colorlinks=true,
urlcolor=blue,
linkcolor=blue,
citecolor=OliveGreen,
]{hyperref}
\fi

\ifnum\showcolorlinks=0
\usepackage[
pagebackref,
colorlinks=false,
pdfborder={0 0 0}
]{hyperref}
\fi

\usepackage{prettyref}

\newcommand{\savehyperref}[2]{\texorpdfstring{\hyperref[#1]{#2}}{#2}}

\newrefformat{eq}{\savehyperref{#1}{\textup{(\ref*{#1})}}}
\newrefformat{eqn}{\savehyperref{#1}{\textup{(\ref*{#1})}}}
\newrefformat{lem}{\savehyperref{#1}{Lemma~\ref*{#1}}}
\newrefformat{def}{\savehyperref{#1}{Definition~\ref*{#1}}}
\newrefformat{thm}{\savehyperref{#1}{Theorem~\ref*{#1}}}
\newrefformat{prog}{\savehyperref{#1}{Program~\ref*{#1}}}
\newrefformat{cor}{\savehyperref{#1}{Corollary~\ref*{#1}}}
\newrefformat{cha}{\savehyperref{#1}{Chapter~\ref*{#1}}}
\newrefformat{sec}{\savehyperref{#1}{Section~\ref*{#1}}}
\newrefformat{app}{\savehyperref{#1}{Appendix~\ref*{#1}}}
\newrefformat{tab}{\savehyperref{#1}{Table~\ref*{#1}}}
\newrefformat{fig}{\savehyperref{#1}{Figure~\ref*{#1}}}
\newrefformat{hyp}{\savehyperref{#1}{Hypothesis~\ref*{#1}}}
\newrefformat{alg}{\savehyperref{#1}{Algorithm~\ref*{#1}}}
\newrefformat{rem}{\savehyperref{#1}{Remark~\ref*{#1}}}
\newrefformat{item}{\savehyperref{#1}{Item~\ref*{#1}}}
\newrefformat{step}{\savehyperref{#1}{step~\ref*{#1}}}
\newrefformat{conj}{\savehyperref{#1}{Conjecture~\ref*{#1}}}
\newrefformat{fact}{\savehyperref{#1}{Fact~\ref*{#1}}}
\newrefformat{prop}{\savehyperref{#1}{Proposition~\ref*{#1}}}
\newrefformat{prob}{\savehyperref{#1}{Problem~\ref*{#1}}}
\newrefformat{claim}{\savehyperref{#1}{Claim~\ref*{#1}}}
\newrefformat{relax}{\savehyperref{#1}{Relaxation~\ref*{#1}}}
\newrefformat{red}{\savehyperref{#1}{Reduction~\ref*{#1}}}
\newrefformat{part}{\savehyperref{#1}{Part~\ref*{#1}}}

\newcommand{\Sref}[1]{\hyperref[#1]{\S\ref*{#1}}}

\usepackage{nicefrac}

\let\nfrac=\nicefrac

\ifnum\usemicrotype=1
\usepackage{microtype}
\fi

\ifnum\showauthornotes=1
\newcommand{\Authornote}[2]{{\sffamily\small\color{red}{[#1: #2]}}}
\newcommand{\Authornotecolored}[3]{{\sffamily\small\color{#1}{[#2: #3]}}}
\newcommand{\Authorcomment}[2]{{\sffamily\small\color{gray}{[#1: #2]}}}
\newcommand{\Authorstartcomment}[1]{\sffamily\small\color{gray}[#1: }

\newcommand{\Authorfnote}[2]{\footnote{\color{red}{#1: #2}}}
\newcommand{\Authorfixme}[1]{\Authornote{#1}{\textbf{??}}}
\newcommand{\Authormarginmark}[1]{\marginpar{\textcolor{red}{\fbox{\Large #1:!}}}}
\else
\newcommand{\Authornote}[2]{}
\newcommand{\Authornotecolored}[3]{}
\newcommand{\Authorcomment}[2]{}
\newcommand{\Authorstartcomment}[1]{}

\newcommand{\Authorfnote}[2]{}
\newcommand{\Authorfixme}[1]{}
\newcommand{\Authormarginmark}[1]{}
\fi

\newcommand{\Tnote}{\Authornotecolored{red}{T}}

\newcommand{\Pnote}{\Authornote{P}}

\definecolor{forestgreen(traditional)}{rgb}{0.0, 0.27, 0.13}

\newcommand{\Snote}{\Authornote{S}}

\ifnum\showfixme=0

\fi

\usepackage{boxedminipage}

\newcommand{\Paren}[1]{\left(#1\right)}

\newcommand{\Brac}[1]{\left[#1\right]}

\newcommand{\norm}[1]{\lVert#1\rVert}
\newcommand{\Norm}[1]{\left\lVert#1\right\rVert}

\newcommand{\iprod}[1]{\langle#1\rangle}
\newcommand{\Iprod}[1]{\left\langle#1\right\rangle}

\newcommand{\Esymb}{\mathbb{E}}
\newcommand{\Psymb}{\mathbb{P}}
\newcommand{\Vsymb}{\mathbb{V}}

\DeclareMathOperator*{\E}{\Esymb}
\DeclareMathOperator*{\Var}{\Vsymb}
\DeclareMathOperator*{\ProbOp}{\Psymb}

\renewcommand{\Pr}{\ProbOp}

\newcommand{\tensor}{\otimes}

\newcommand{\textparen}[1]{\text{(#1)}}

\ifx\because\undefined
\newcommand{\because}[1]{\textparen{because #1}}
\else
\renewcommand{\because}[1]{\textparen{because #1}}
\fi

\newcommand{\sge}{\succeq}

\newcommand{\bits}{\{0,1\}}

\newcommand{\defeq}{\stackrel{\mathrm{def}}=}

\newcommand{\mper}{\,.}

\newcommand\bdot\bullet

\ifx\mathds\undefined 
\DeclareMathOperator{\Ind}{\mathbb{I}}
\else
\DeclareMathOperator{\Ind}{\mathds 1}}
\fi

\DeclareMathOperator{\Tr}{Tr}

\DeclareMathOperator{\opt}{opt}

\DeclareMathOperator{\sign}{sign}

\newcommand{\Erdos}{Erd\H{o}s\xspace}

\newcommand{\N}{\mathbb N}
\newcommand{\R}{\mathbb R}

\newcommand{\problemmacro}[1]{\texorpdfstring{\textup{\textsc{#1}}}{#1}\xspace}

\newcommand{\maxkcsp}{\problemmacro{max $k$-csp}}

\newcommand{\maxclique}{\problemmacro{max clique}}
\newcommand{\densestksubgraph}{\problemmacro{densest $k$-subgraph}}
\newcommand{\sparsepca}{\problemmacro{sparse PCA}}
\newcommand{\tensorpca}{\problemmacro{tensor PCA}}
\newcommand{\communitydetection}{\problemmacro{community detection}}

\newcommand{\cA}{\mathcal A}
\newcommand{\cB}{\mathcal B}

\newcommand{\cD}{\mathcal D}

\newcommand{\cG}{\mathcal G}

\newcommand{\cI}{\mathcal I}

\newcommand{\cL}{\mathcal L}

\newcommand{\cN}{\mathcal N}

\newcommand{\cP}{\mathcal P}

\newcommand{\cR}{\mathcal R}
\newcommand{\cS}{\mathcal S}

\newcommand{\cX}{\mathcal X}

\newcommand{\scrI}{\mathscr I}

\newcommand{\bbS}{\mathbb S}

\renewcommand{\leq}{\leqslant}
\renewcommand{\le}{\leqslant}
\renewcommand{\geq}{\geqslant}
\renewcommand{\ge}{\geqslant}

\ifnum\showdraftbox=1
\newcommand{\draftbox}{\begin{center}
  \fbox{%
    \begin{minipage}{2in}%
      \begin{center}%
          \Large\textsc{Working Draft}\\%
        Please do not distribute%
      \end{center}%
    \end{minipage}%
  }%
\end{center}
\vspace{0.2cm}}
\else
\newcommand{\draftbox}{}
\fi

\let\epsilon=\varepsilon

\numberwithin{equation}{section}

\newcommand\MYcurrentlabel{xxx}

\newcommand{\MYstore}[2]{%
  \global\expandafter \def \csname MYMEMORY #1 \endcsname{#2}%
}

\newcommand{\MYload}[1]{%
  \csname MYMEMORY #1 \endcsname%
}

\newcommand{\MYnewlabel}[1]{%
  \renewcommand\MYcurrentlabel{#1}%
  \MYoldlabel{#1}%
}

\newcommand{\MYdummylabel}[1]{}

\newcommand{\torestate}[1]{%
  \let\MYoldlabel\label%
  \let\label\MYnewlabel%
  #1%
  \MYstore{\MYcurrentlabel}{#1}%
  \let\label\MYoldlabel%
}

\newcommand{\restatetheorem}[1]{%
  \let\MYoldlabel\label
  \let\label\MYdummylabel
  \begin{theorem*}[Restatement of \prettyref{#1}]
    \MYload{#1}
  \end{theorem*}
  \let\label\MYoldlabel
}

\newcommand{\restatelemma}[1]{%
  \let\MYoldlabel\label
  \let\label\MYdummylabel
  \begin{lemma*}[Restatement of \prettyref{#1}]
    \MYload{#1}
  \end{lemma*}
  \let\label\MYoldlabel
}

\newcommand{\restateprop}[1]{%
  \let\MYoldlabel\label
  \let\label\MYdummylabel
  \begin{proposition*}[Restatement of \prettyref{#1}]
    \MYload{#1}
  \end{proposition*}
  \let\label\MYoldlabel
}

\newcommand{\restatefact}[1]{%
  \let\MYoldlabel\label
  \let\label\MYdummylabel
  \begin{fact*}[Restatement of \prettyref{#1}]
    \MYload{#1}
  \end{fact*}
  \let\label\MYoldlabel
}

\newcommand{\restate}[1]{%
  \let\MYoldlabel\label
  \let\label\MYdummylabel
  \MYload{#1}
  \let\label\MYoldlabel
}

\newcommand{\addreferencesection}{
  \phantomsection
  \addcontentsline{toc}{section}{References}
}

\newcommand{\e}{\epsilon}
\newcommand{\eps}{\epsilon}

\let\origparagraph\paragraph
\renewcommand{\paragraph}[1]{\origparagraph{#1.}}

\allowdisplaybreaks

\usepackage{paralist}

\usepackage{comment}

\usepackage{relsize}
\usepackage[font=footnotesize]{caption}

\DeclareMathOperator{\Span}{span}

\DeclareMathOperator{\obj}{obj}

\DeclareMathOperator{\Id}{\mathrm{Id}}

\DeclareUrlCommand\email{}

\DeclareMathOperator*{\pE}{\widetilde{\mathbb E}}

\newcommand{\bT}{\mathbf{T}}

\let\pref=\prettyref

\newcommand*{\dyad}[1]{#1#1{}^{\mkern-4mu\intercal}}

\newcommand{\sdpopt}{\mathop{\textrm{sdpOpt}}}

\newcommand{\sos}{SoS\xspace}

\newcommand{\pdist}{\mu}
\newcommand{\udist}{\nu}
\newcommand{\instN}{N}
\newcommand{\sol}{x}
\newcommand{\sols}{\cX}
\newcommand{\success}{s}
\newcommand{\mempty}{\mathbf{e}_{\emptyset,\emptyset}}

\newcommand{\plantedclique}{\textsc{Planted Clique}}

\newcommand{\subinst}{\downarrow}
\newcommand{\subsetdist}{\Theta}
\newcommand{\inst}{\mathrm{inst}}
\newcommand{\prog}{\mathrm{prog}}

\renewcommand{\pE}{\tilde{\mathbb{E}}}

\title{The power of sum-of-squares for detecting hidden structures}

\author{%
  Samuel B. Hopkins\thanks{Cornell University, \protect\email{samhop@cs.cornell.edu} Partially supported by an NSF GRFP under grant no. 1144153, by a Microsoft Research Graduate Fellowship, and by David Steurer's NSF CAREER award.}
\and
Pravesh K. Kothari \thanks{Princeton University and IAS, \protect \email{kothari@cs.princeton.edu}}
\and
Aaron Potechin
\and
Prasad Raghavendra
\and
Tselil Schramm\thanks{UC Berkeley, \protect\email{tscrhamm@cs.berkeley.edu}. Supported by an NSF Graduate Research Fellowship (1106400).}
\and
David Steurer\thanks{Cornell University, \protect\email{dsteurer@cs.cornell.edu}.
Supported by a Microsoft Research Fellowship, a Alfred P. Sloan Fellowship, an NSF CAREER award, and the Simons Collaboration for Algorithms and Geometry.}}

\begin{document}

\maketitle
\draftbox
\thispagestyle{empty}

\begin{abstract}
  We study planted problems---finding hidden structures in random noisy inputs---through the lens of the sum-of-squares semidefinite programming hierarchy (\sos).
  This family of powerful semidefinite programs has recently yielded many new algorithms for planted problems, often achieving the best known polynomial-time guarantees in terms of accuracy of recovered solutions and robustness to noise.
  One theme in recent work is the design of spectral algorithms which match the guarantees of \sos algorithms for planted problems.
  Classical spectral algorithms are often unable to accomplish this: the twist in these new spectral algorithms is the use of spectral structure of matrices whose entries are low-degree polynomials of the input variables.

  We prove that for a wide class of planted problems, including refuting random constraint satisfaction problems, tensor and sparse PCA, densest-$k$-subgraph, community detection in stochastic block models, planted clique, and others, eigenvalues of degree-$d$ matrix polynomials are as powerful as \sos semidefinite programs of size roughly $n^d$.
  For such problems it is therefore always possible to match the guarantees of \sos without solving a large semidefinite program.

  Using related ideas on \sos algorithms and low-degree matrix polynomials (and inspired by recent work on \sos and the planted clique problem \cite{DBLP:conf/focs/BarakHKKMP16}), we prove new nearly-tight \sos lower bounds for the tensor and sparse principal component analysis problems.
  Our lower bounds are the first to suggest that improving upon the signal-to-noise ratios handled by existing polynomial-time algorithms for these problems may require subexponential time.
\end{abstract}

\clearpage

\ifnum\showtableofcontents=1
{
\tableofcontents
\thispagestyle{empty}
 }
\fi

\clearpage

\setcounter{page}{1}

\section{Introduction}
\label{sec:introduction}

Recent years have seen a surge of progress in algorithm design via the sum-of-squares (\sos) semidefinite programming hierarchy.
Initiated by the work of \cite{DBLP:conf/stoc/BarakBHKSZ12}, who showed that polynomial time algorithms in the hierarchy solve all known integrality gap instances for Unique Games and related problems, a steady stream of works have developed efficient algorithms for both worst-case  \cite{DBLP:conf/stoc/BarakKS14,DBLP:conf/stoc/BarakKS15,DBLP:journals/corr/BarakKS17,DBLP:journals/eccc/BhattiproluGGLT16} and average-case problems \cite{DBLP:conf/colt/HopkinsSS15,DBLP:conf/approx/GeM15,DBLP:conf/colt/BarakM16,DBLP:journals/corr/RaghavendraRS16,DBLP:journals/corr/BhattiproluGL16,DBLP:journals/corr/MaSS16,DBLP:journals/corr/PotechinS17}.
The insights from these works extend beyond individual algorithms to characterizations of broad classes of algorithmic techniques.
In addition, for a large class of problems (including constraint satisfaction), the family of SoS semidefinite programs is now known to be as powerful as \emph{any} semidefinite program (SDP) \cite{DBLP:conf/stoc/LeeRS15}.

In this paper we focus on recent progress in using Sum of Squares algorithms to solve average-case, and especially \emph{planted} problems---problems that ask for the recovery of a planted \emph{signal} perturbed by random \emph{noise}.
Key examples are finding solutions of random constraint satisfaction problems (CSPs) with planted assignments \cite{DBLP:journals/corr/RaghavendraRS16} and finding planted optima of random polynomials over the $n$-dimensional unit sphere \cite{DBLP:journals/corr/RaghavendraRS16,DBLP:journals/corr/BhattiproluGL16}.
The latter formulation captures a wide range of unsupervised learning problems, and has led to many unsupervised learning algorithms with the best-known polynomial time guarantees \cite{DBLP:conf/stoc/BarakKS15, DBLP:conf/stoc/BarakKS14, DBLP:conf/focs/MaSS16, DBLP:conf/colt/HopkinsSS15, DBLP:journals/corr/PotechinS17, DBLP:journals/eccc/BhattiproluGGLT16}.

In many cases, classical algorithms for such planted problems are \emph{spectral} algorithms---i.e., using the top eigenvector of a natural matrix associated with the problem input to recover a planted solution. The canonical algorithms for the \emph{planted clique}  \cite{DBLP:journals/rsa/AlonKS98}, \emph{principal components analysis} (PCA) \cite{Pearson1901}, and \emph{tensor decomposition} (which is intimately connected to optimizaton of polynomials on the unit sphere) \cite{harshman1970foundations} are all based on this general scheme.
In all of these cases, the algorithm employs the top eigenvector of a matrix which is either given as input (the adjacency matrix, for planted clique), or is a simple function of the input (the empirical covariance, for PCA).

Recent works have shown that one can often improve upon these basic spectral methods using \sos, yielding better accuracy and robustness guarantees against noise in recovering planted solutions.
Furthermore, for worst case problems---as opposed to the average-case planted problems we consider here---semidefinite programs are strictly more powerful than spectral algorithms.\footnote{For example, consider the contrast between the SDP algorithm for Max-Cut of Goemans and Williamson, \cite{DBLP:conf/stoc/GoemansW94}, and the spectral algorithm of Trevisan  \cite{DBLP:conf/stoc/Trevisan09}; or the SDP-based algorithms for coloring worst-case 3-colorable graphs \cite{Kawa} relative to the best spectral methods \cite{DBLP:journals/siamcomp/AlonK97} which only work for random inputs.}
\emph{A priori} one might therefore expect that these new \sos guarantees for planted problems would not be achievable via spectral algorithms.
But curiously enough, in numerous cases these stronger guarantees for planted problems can be achieved by spectral methods!
The twist is that the entries of these matrices are low-degree polynomials in the input to the algorithm .
The result is a new family of low-degree spectral algorithms with guarantees matching \sos but requriring only eigenvector computations instead of general semidefinite programming \cite{DBLP:conf/stoc/HopkinsSSS16,DBLP:journals/corr/RaghavendraRS16,MR3473335-Allen15}.

This leads to the following question which is the main focus of this work.

\centerline{ \emph{Are \sos algorithms equivalent to low-degree spectral methods
for planted problems?}}

We answer this question affirmatively for a wide class of distinguishing problems which includes refuting random CSPs, tensor and sparse PCA, densest-$k$-subgraph, community detection in stochastic block models, planted clique, and more. Our positive answer to this question implies that a light-weight algorithm---computing the top eigenvalue of a single matrix whose entries are low-degree polynomials in the input---can recover the performance guarantees of an often bulky semidefinite programming relaxation.

To complement this picture, we prove two new \sos lower bounds for particular planted problems, both variants of component analysis: sparse principal component analysis and tensor principal component analysis (henceforth sparse PCA and tensor PCA, respectively) \cite{spca-original,DBLP:conf/nips/RichardM14}.
For both problems there are nontrivial low-degree spectral algorithms, which have better noise tolerance than naive spectral methods \cite{DBLP:conf/stoc/HopkinsSSS16,DBLP:conf/nips/DeshpandeM14, DBLP:journals/corr/RaghavendraRS16, DBLP:journals/corr/BhattiproluGL16}.
Sparse PCA, which is used in machine learning and statistics to find important coordinates in high-dimensional data sets, has attracted much attention in recent years for being apparently computationally intractable to solve with a number of samples which is more than sufficient for brute-force algorithms \cite{krauthgamer2015semidefinite,berthet2013computational,DBLP:conf/nips/MaW15}.
Tensor PCA appears to exhibit similar behavior \cite{DBLP:conf/colt/HopkinsSS15}.
That is, both problems exhibit \emph{information-computation gaps}.

Our \sos lower bounds for both problems are the strongest yet formal evidence for information-computation gaps for these problems.
We rule out the possibility of subexponential-time \sos algorithms which improve by polynomial factors on the signal-to-noise ratios tolerated by the known low degree spectral methods.
In particular, in the case of sparse PCA, it appeared possible prior to this work that it might be possible in quasipolynomial time to recover a $k$-sparse unit vector $v$ in $p$ dimensions from $O(k \log p)$ samples from the distribution $\cN(0, \Id + vv^\top)$.
Our lower bounds suggest that this is extremely unlikely; in fact this task probably requires polynomial \sos degree and hence $\exp( n^{\Omega(1)})$ time for \sos algorithms.
This demonstrates that (at least with regard to \sos algorithms) both problems are much harder than the \emph{planted clique} problem, previously used as a basis for reductions in the setting of sparse PCA \cite{berthet2013computational}.

Our lower bounds for sparse and tensor PCA are closely connected to the failure of low-degree spectral methods in high noise regimes of both problems.
We prove them both by showing that with noise beyond what known low-degree spectral algorithms can tolerate, even low-degree \emph{scalar} algorithms (the result of restricting low-degree spectral algorithms to $1 \times 1$ matrices) would require subexponential time to detect and recover planted signals.
We then show that in the restricted settings of tensor and sparse PCA, ruling out these weakened low-degree spectral algorithms is enough to imply a strong \sos lower bound.

\subsection{SoS and spectral algorithms for robust inference}
We turn to our characterization of \sos algorithms for planted problems in terms of low-degree spectral algorithms.
First, a word on planted problems.
Many planted problems have several formulations: \emph{search}, in which the goal is to recover a planted solution, \emph{refutation}, in which the goal is to certify that no planted solution is present, and \emph{distinguishing}, where the goal is to determine with good probability whether an instance contains a planted solution or not.
Often an algorithm for one version can be parlayed into algorithms for the others, but distinguishing problems are often the easiest, and we focus on them here.

A distinguishing problem is specified by two distributions on instances: a \emph{planted} distribution supported on instances with a hidden structure, and a \emph{uniform} distribution, where samples w.h.p. contain no hidden structure.
Given an instance drawn with equal probability from the planted or the uniform distribution, the goal is to determine with probability greater than $\tfrac 1 2$ whether or not the instance comes from the planted distribution.
For example:

\textbf{Planted clique } \emph{Uniform distribution: } $G(n,\tfrac 1 2)$, the \Erdos-Renyi distribution, which w.h.p. contains no clique of size $\omega(\log n)$.
\emph{Planted distribution: } The uniform distribution on graphs containing a $n^{\e}$-size clique, for some $\e > 0$. (The problem gets harder as $\e$ gets smaller, since the distance between the distributions shrinks.)

\textbf{Planted $3$\textsc{xor} }
\emph{Uniform distribution: } a $3$\textsc{xor} instance on $n$ variables and $m > n$ equations $x_i x_j x_k = a_{ijk}$, where all the triples $(i,j,k)$ and the signs $a_{ijk} \in \{ \pm 1\}$ are sampled uniformly and independently. No assignment to $x$ will satisfy more than a $0.51$-fraction of the equations, w.h.p.
  \emph{Planted distribution: } The same, except the signs $a_{ijk}$ are sampled to correlate with $b_i b_j b_k$ for a randomly chosen $b_i \in \{ \pm 1\}$, so that the assignment $x = b$ satisfies a $0.9$-fraction of the equations.
  (The problem gets easier as $m/n$ gets larger, and the contradictions in the uniform case become more locally apparent.)

  We now formally define a family of distinguishing problems, in order to give our main theorem.
Let $\scrI$ be a set of instances corresponding to a product space (for concreteness one may think of $\scrI$ to be the set of graphs on $n$ vertices, indexed by $\{0,1\}^{\binom{n}{2}}$, although the theorem applies more broadly).
Let $\nu$, our uniform distrbution, be a product distribution on $\scrI$.

With some decision problem $\cP$ in mind (e.g. does $G$ contain a clique of size $\ge n^\eps$?), let $\cX$ be a set of solutions to $\cP$; again for concreteness one may think of $\cX$ as being associated with cliques in a graph, so that $\cX \subset \{0,1\}^n$ is the set of all indicator vectors on at least $n^{\eps}$ vertices.

For each solution $x \in \cX$, let $\mu_{|_x}$ be the uniform distribution over instances $I \in \scrI$ that contain $x$.
For example, in the context of planted clique, if $x$ is a clique on vertices $1,\ldots,n^{\e}$, then $\mu_{|_x}$ would be the uniform distribution on graphs containing the clique $1,\ldots,n^{\e}$.
We define the planted distribution $\mu$ to be the uniform mixture over $\mu_x$, $\mu = U_{x \sim \cX} \mu_{|_x}$.

The following is our main theorem on the equivalence of sum of squares algorithms for distinguishing problems and spectral algorithms employing low-degree matrix polynomials.

\begin{theorem}[Informal]\label{thm:maindist}
    Let $N,n \in \cN$, and let $\cA,\cB$ be sets of real numbers.
    Let $\scrI$ be a family of instances over $\cA^N$, and let $\cP$ be a decision problem over $\scrI$ with $\cX = \cB^n$ the set of possible solutions to $\cP$ over $\scrI$.
  Let $\{g_j(x,I)\}$ be a system of $n^{O(d)}$ polynomials of degree at most $d$ in the variables $x$ and constant degree in the variables $\cI$ that encodes $\cP$, so that
  \begin{itemize}
    \item for $I \sim_{\nu} \scrI$, with high probability the system is unsatisfiable and admits a degree-$d$ SoS refutation, and
    \item for $I \sim_{\mu} \scrI$, with high probability the system is satisfiable by some solution $x \in X$, and $x$ remains feasible even if all but an $n^{-0.01}$-fraction of the coordinates of $\cI$ are re-randomized according to $\nu$.
  \end{itemize}
  Then there exists a matrix whose entries are degree-$O(d)$ polynomials $Q : \scrI \rightarrow \R^{{n \choose \leq d} \times {n \choose \leq d}}$ such that
  \[
    \E_{I \sim \nu} \Brac{\lambda_{max}^+(Q(I))} \leq 1,\quad \text{ while } \quad \E_{I \sim \mu} \Brac{\lambda_{max}^+(Q(I))} \geq n^{10d},
  \]
  where $\lambda^+_{\max}$ denotes the maximum non-negative eigenvalue.
\end{theorem}

The condition that a solution $x$ remain feasible if all but a fraction of the coordinates of $I \sim \mu_{|_x}$ are re-randomized should be interpreted as a noise-robustness condition.
To see an example, in the context of planted clique, suppose we start with a planted distribution over graphs with a clique $x$ of size $n^{\eps+0.01}$.
If a random subset of $n^{0.99}$ vertices are chosen, and all edges not entirely contained in that subset are re-randomized according to the $G(n,1/2)$ distribution, then with high probability at least $n^{\eps}$ of the vertices in $x$ remain in a clique, and so $x$ remains feasible for the problem $\cP$: $G$ has a clique of size $\ge n^{\eps}$?

\subsection{SoS and information-computation gaps}
Computational complexity of planted problems has become a rich area of study.
The goal is to understand which planted problems admit efficient (polynomial time) algorithms, and to study the \emph{information-computation gap} phenomenon: many problems have noisy regimes in which planted structures can be found by inefficient algorithms, but (conjecturally) not by polynomial time algorithms.
One example is the \emph{planted clique} problem, where the goal find a large clique in a sample from the uniform distribution over graphs containing a clique of size $n^\eps$ for a small constant $\eps>0$.
While the problem is solvable for any $\e > 0$ by a brute-force algorithm requiring $n^{\Omega(\log n)}$ time, polynomial time algorithms are conjectured to require $\e \geq \tfrac 12$.

A common strategy to provide evidence for such a gap is to prove that powerful classes of efficient algorithms are unable to solve the planted problem in the (conjecturally) hard regime.
SoS algorithms are particularly attractive targets for such lower bounds because of their broad applicability and strong guarantees.

In a recent work, Barak et al. \cite{DBLP:conf/focs/BarakHKKMP16} show an SoS lower bound for the planted clique problem, demonstrating that when $\e < \tfrac 12$, SoS algorithms require $n^{\Omega(\log n)}$ time to solve planted clique.
Intriguingly, they show that in the case of planted clique that \sos algorithms requiring $\approx n^d$ time can distinguish planted from random graphs only when there is a \emph{scalar-valued} degree $\approx d \cdot \log n$ polynomial $p(A) \, : \R^{n \times n} \rightarrow \R$ (here $A$ is the adjacency matrix of a graph) with
\[
  \E_{G(n,1/2)} p(A) = 0, \quad \E_{\text{planted}} p(A) \geq n^{\Omega(1)} \cdot \Paren{\Var_{G(n,1/2)} p(A)}^{1/2}\mper
\]
That is, such a polynomial $p$ has much larger expectation in under the planted distribution than its standard deviation in uniform distribution.
(The choice of $n^{\Omega(1)}$ is somewhat arbitrary, and could be replaced with $\Omega(1)$ or $n^{\Omega(d)}$ with small changes in the parameters.)
By showing that as long as $\eps < \frac{1}{2}$ any such polynomial $p$ must have degree $\Omega(\log n)^2$, they rule out efficient \sos algorithms when $\eps < \frac{1}{2}$.
Interestingly, this matches the spectral distinguishing threshold---the spectral algorithm of \cite{DBLP:journals/rsa/AlonKS98} is known to work when $\eps \ge \frac{1}{2}$.

This stronger characterization of \sos for the planted clique problem, in terms of \emph{scalar} distinguishing algorithms rather than \emph{spectral} distinguishing algorihtms, may at first seem insignificant.
To see why the scalar characterization is more powerful, we point out that if the degree-$d$ moments of the planted and uniform distributions are known, determining the optimal scalar distinguishing polynomial is easy: given a planted distribution $\pdist$ and a random distribution $\udist$ over instances $\cI$, one just solves a linear algebra problem in the $n^{d\log n}$ coefficients of $p$ to maximize the expectation over $\pdist$ relative to $\udist$:
\[
    \max_{p} \E_{\cI \sim \pdist} [p^2(\cI)] \quad  s.t. \ \E_{\cI \sim \udist}[p^2(\cI)] = 1\mper
\]
It is not difficult to show that the optimal solution to the above program has a simple form: it is the projection of the \emph{relative density of $\nu$ with respect to $\mu$} projected to the degree-$d\log n$ polynomials. \Tnote{}
So given a pair of distributions $\pdist,\udist$, in $n^{O(d\log n)}$ time, it is possible to determine whether there exists a degree-$d\log n$ scalar distinguishing polynomial.
Answering the same question about the existence of a spectral distinguisher is more complex, and to the best of our knowledge cannot be done efficiently.

Given this powerful theorem for the case of the planted clique problem, one may be tempted to conjecture that this stronger, \emph{scalar} distinguisher characterization of the \sos algorithm applies more broadly than just to the planted clique problem, and perhaps as broadly as \pref{thm:maindist}.
If this conjecture is true, given a pair of distributions $\udist$ and $\pdist$ with known moments, it would be possible in many cases to efficiently and mechanically determine whether polynomial-time \sos distinguishing algorithms exist!
\begin{conjecture}
  \label{conj:main-conjecture}
  In the setting of \pref{thm:maindist}, the conclusion may be replaced with the conclusion that there exists a scalar-valued polynomial $p : \scrI \rightarrow \R$ of degree $O(d \cdot \log n)$ so that
  \[
    \E_{\text{uniform}} p(I) = 0 \text{ and } \E_{\text{planted}} p(I) \geq n^{\Omega(1)} \Paren{\E_{\text{uniform}} p(I)^2 } ^{1/2}
  \]
\end{conjecture}

To illustrate the power of this conjecture, in the beginning of Section~\ref{sec:pca} we give a short and self-contained explanation of how this predicts, via simple linear algebra, our $n^{\Omega(1)}$-degree \sos lower bound for tensor PCA.
As evidence for the conjecture, we verify this prediction by proving such a lower bound unconditionally.

We also note why \pref{thm:maindist} does not imply Conjecture~\ref{conj:main-conjecture}.
While, in the notation of that theorem, the entries of $Q(I)$ are low-degree polynomials in $I$, the function $M \mapsto \lambda^+_{\max}(M)$ is not (to the best of our knowledge) a low-degree polynomial in the entries of $M$ (even approximately).
(This stands in contrast to, say the operator norm or Frobenious norm of $M$, both of which are exactly or approximately low-degree polynomials in the entries of $M$.)
This means that the final output of the spectral distinguishing algorithm offered by \pref{thm:maindist} is not a low-degree polynomial in the instance $I$.

\subsection{Exponential lower bounds for sparse PCA and tensor PCA}
Our other main results are strong exponential lower bound on the sum-of-squares method (specifically, against $2^{n^{\Omega(1)}}$ time or $n^{\Omega(1)}$ degree algorithms) for the tensor and sparse principal component analysis (PCA).
We prove the lower bounds by extending the techniques pioneered in \cite{DBLP:conf/focs/BarakHKKMP16}.
In the present work we describe the proofs informally, leaving full details to a forthcoming full version.

\paragraph{Tensor PCA}
We start with the simpler case of tensor PCA, introduced by \cite{DBLP:conf/nips/RichardM14}.
\begin{problem}[Tensor PCA]
    Given an order-$k$ tensor in $(\R^{n})^{\tensor k}$, determine whether it comes from:
    \begin{compactitem}
    \item {\bf Uniform Distribution}: each entry of the tensor sampled independently from $\cN(0,1)$.
    \item {\bf Planted Distribution}: a spiked tensor, $\bT = \lambda \cdot v^{\tensor k} + G$  where $v$ is sampled uniformly from $\bbS^{n-1}$, and where $G$ is a random tensor with  each entry sampled independently from $\cN(0,1)$.
    \end{compactitem}
\end{problem}
Here, we think of $v$ as a signal hidden by Gaussian noise.
The parameter $\lambda$ is a signal-to-noise ratio.
In particular, as $\lambda$ grows, we expect the distinguishing problem above to get easier.

Tensor PCA is a natural generalization of the PCA problem in machine learning and statistics.
Tensor methods in general are useful when data naturally has more than two modalities: for example, one might consider a recommender system which factors in not only people and movies but also time of day.
Many natural tensor problems are NP hard in the worst-case.
Though this is not necessarily an obstacle to machine learning applications, it is important to have average-case models to in which to study algorithms for tensor problems.
The spiked tensor setting we consider here is one such simple model.

Turning to algorithms: consider first the ordinary PCA problem in a spiked-matrix model.
Given an $n \times n$ matrix $M$, the problem is to distinguish between the case where every entry of $M$ is independently drawn from the standard Gaussian distribution $\cN(0,1)$ and the case when $M$ is drawn from a distribution as above with an added rank one shift $\lambda vv^{\top}$ in a uniformly random direction $v$.
A natural and well-studied algorithm, which solves this problem to information-theoretic optimality is to threshold on the largest singular value/spectral norm of the input matrix.
Equivalently, one thresholds on the maximizer of the degree two polynomial $\langle x, M x \rangle$ in $x \in \bbS^{n-1}.$

A natural generalization of this algorithm to the tensor PCA setting (restricting for simplicity $k = 3$ for this discussion) is the maximum of the degree-three  polynomial $\langle T, x^{\otimes 3} \rangle$ over the unit sphere---equivalently, the (symmetric) injective tensor norm of $T$.
This maximum can be shown to be much larger in case of the planted distribution so long as $\lambda \gg \sqrt{n}$.
Indeed, this approach to distinguishing between planted and uniform distributions is information-theoretically optimal \cite{DBLP:journals/corr/PerryWB16, DBLP:journals/corr/BanksMVX16}.
Since recovering the spike $v$ and optimizing the polynomial $\iprod{T,x^{\tensor 3}}$ on the sphere are equivalent, tensor PCA can be thought of as an average-case version of the problem of optimizing a degree-$3$ polynomial on the unit sphere (this problem is NP hard in the worst case, even to approximate \cite{DBLP:journals/corr/abs-0911-1393, DBLP:conf/stoc/BarakBHKSZ12}).

Even in this average-case model, it is believed that there is a gap between which signal strengths $\lambda$ allow recovery of $v$ by brute-force methods and which permit polynomial time algorithms.
This is quite distinct from the vanilla PCA setting, where eigenvector algorithms solve the spike-recovery problem to information-theoretic optimality.
Nevertheless, the best-known algorithms for tensor PCA arise from computing convex relaxations of this degree-$3$ polynomial optimization problem.
Specifically, the \sos method captures the state of the art algorithms for the problem; it is known to recover the vector $v$ to $o(1)$ error in polynomial time whenever $\lambda \gg n^{3/4}$ \cite{DBLP:conf/colt/HopkinsSS15}.
A major open question in this direction is to understand the complexity of the problem for $\lambda \leq n^{3/4 - \e}$.
Algorithms (again captured by \sos) are known which run in $2^{n^{O(\e)}}$ time \cite{DBLP:journals/corr/RaghavendraRS16, DBLP:journals/eccc/BhattiproluGGLT16}.
We show the following theorem which shows that the sub-exponential algorithm above is in fact nearly optimal for \sos algorithm.

\begin{theorem}\label{thm:tpca-intro}
  For a tensor $T$, let
  \[
    \sos_d(T) = \max_{\pE} \pE[ \iprod{T, x^{\tensor k}}] \text{ such that $\pE$ is a degree $d$ pseudoexpectation and satisfies } \{ \|x\|^2 = 1 \}\footnote{For definitions of pseudoexpectations and related matters, see the survey \cite{DBLP:journals/corr/BarakS14}. }
  \]
For every small enough constant $\e > 0$, if $T \in \R^{n \times n \times n}$ has iid Gaussian or $\{ \pm 1\}$ entries,
  $\E_T \sos_d(T) \geq n^{k/4 - \e}$, for every $d \leq n^{c \cdot \e}$ for some universal $c > 0$.
\end{theorem}
In particular for third order tensors (i.e $k = 3$), since degree $n^{\Omega(\e)}$ \sos is unable to certify that a random $3$-tensor has maximum value much less than $n^{3/4 - \e}$, this \sos relaxation cannot be used to distinguish the planted and random distributions above when $\lambda \ll n^{3/4 - \e}$.\footnote{In fact, our proof for this theorem will show somewhat more: that a large family of constraints---any valid constraint which is itself a low-degree polynomial of $T$---could be added to this convex relaxation and the lower bound would still obtain.}

\paragraph{Sparse PCA}
We turn to sparse PCA, which we formalize as the following planted distinguishing problem.
\begin{problem}[Sparse PCA $(\lambda,k)$]
    Given an $n \times n$ symmetric real matrix $A$, determine whether $A$ comes from:
    \begin{compactitem}
    \item {\bf Uniform Distribution}: each upper-triangular entry of the matrix $A$ is sampled iid from $\cN(0,1)$; other entries are filled in to preserve symmetry.
    \item {\bf Planted Distribution}: a random $k$-sparse unit vector $v$ with entries $\{ \pm 1/\sqrt k, 0\}$ is sampled, and $B$ is sampled from the uniform distribution above; then $A = B + \lambda \cdot \dyad{v}$.
    \end{compactitem}
\end{problem}

We defer significant discussion to Section~\ref{sec:pca}, noting just a few things before stating our main theorem on sparse PCA.
First, the planted model above is sometimes called the \emph{spiked Wigner} model---this refers to the independence of the entries of the matrix $B$.
An alternative model for sparse PCA is the \emph{spiked Wishart} model: $A$ is replaced by $\sum_{i \leq m} \dyad{x_i}$, where each $x_i \sim \cN(0, \Id + \beta \dyad{v})$, for some number $m \in \N$ of samples and some signal-strength $\beta \in \R$.
Though there are technical differences between the models, to the best of our knowledge all known algorithms with provable guarantees are equally applicable to either model; we expect that our \sos lower bounds also apply in the spiked Wishart model.

We generally think of $k,\lambda$ as small powers of $n$; i.e. $n^\rho$ for some $\rho \in (0,1)$; this allows us to generally ignore logarithmic factors in our arguments.
As in the tensor PCA setting, a natural and information-theoretically optimal algorithm for sparse PCA is to maximize the quadratic form $\iprod{x,Ax}$, this time over $k$-sparse unit vectors.
For $A$ from the uniform distribution standard techniques ($\e$-nets and union bounds) show that the maximum value achievable is $O(\sqrt k \log n)$ with high probability, while for $A$ from the planted model of course $\iprod{v,Av} \approx \lambda$.
So, when $\lambda \gg \sqrt k$ one may distinguish the two models by this maximum value.

However, this maximization problem is NP hard for general quadratic forms $A$ \cite{DBLP:conf/colt/ChanPR16}.
So, efficient algorithms must use some other distinguisher which leverages the randomness in the instances.
Essentially only two polynomial-time-computable distinguishers are known.\footnote{If one studies the problem at much finer granularity than we do here, in particular studying $\lambda$ up to low-order additive terms and how precisely it is possible to estimate the planted signal $v$, then the situation is more subtle \cite{DBLP:conf/isit/DeshpandeM14}.}
If $\lambda \gg \sqrt n$ then the maximum eigenvalue of $A$ distinguishes the models.
If $\lambda \gg k$ then the planted model can be distinguished by the presence of large diagonal entries of $A$.
Notice both of these distinguishers fail for some choices of $\lambda$ (that is, $\sqrt k \ll \lambda \ll \sqrt n, k$) for which brute-force methods (optimizing $\iprod{x,Ax}$ over sparse $x$) could successfully distinguish planted from uniform $A$'s.
The theorem below should be interpreted as an impossibility result for \sos algorithms in the $\sqrt k \ll \lambda \ll \sqrt n, k$ regime.
This is the strongest known impossibility result for sparse PCA among those ruling out classes of efficient algorithms (one reduction-based result is also know, which shows sparse PCA is at least as hard as the planted clique problem \cite{DBLP:conf/colt/BerthetR13}.
It is also the first evidence that the problem may require subexponential (as opposed to merely quasi-polynomial) time.

\begin{theorem}\label{thm:spca-main}
  If $A \in \R^{n \times n}$, let
  \[
    SoS_{d,k}(A) = \max_{\pE} \pE \iprod{x,Ax} \text{ s.t. $\pE$ is degree $d$ and satisfies }\left \{ x_i^3 = x_i, \|x\|^2 = k \right \} \mper
  \]
  There are absolute constants $c,\e^* > 0$ so that for every $\rho \in (0,1)$ and $\e \in (0,\e^*)$, if $k = n^{\rho}$, then for $d \leq n^{c \cdot \e}$,
  \[
    \E_{A \sim \{\pm 1\}^{\binom{n}{2}}} SoS_{d,k}(A) \geq \min (n^{1/2 - \epsilon} k, n^{\rho - \e} k)\mper
  \]
\end{theorem}

For more thorough discussion of the theorem, see Section~\ref{sec:spca-main}.

\subsection{Related work}
\paragraph{On interplay of \sos relaxations and spectral methods}
As we have already alluded to, many prior works explore the connection between \sos relaxations and spectral algorithms, beginning with the work of \cite{DBLP:conf/stoc/BarakBHKSZ12} and including the followup works \cite{DBLP:conf/colt/HopkinsSS15,DBLP:conf/focs/AllenOW15,DBLP:conf/colt/BarakM16} (plus many more).
Of particular interest are the papers \cite{DBLP:conf/stoc/HopkinsSSS16,DBLP:journals/corr/MontanariS16}, which use the \sos algorithms to obtain \emph{fast} spectral algorithms, in some cases running in time linear in the input size (smaller even than the number of variables in the associated \sos SDP).

In light of our \pref{thm:maindist}, it is particularly interesting to note cases in which the known \sos lower bounds matching the known spectral algorithms---these problems include planted clique (upper bound: \cite{DBLP:journals/rsa/AlonKS98}, lower bound:\footnote{SDP lower bounds for the planted clique problem were known for smaller degrees of sum-of-squares relaxations and for other SDP relaxations before; see the references therein for details.} \cite{DBLP:conf/focs/BarakHKKMP16}), strong refutations for random CSPs (upper bound:\footnote{There is a long line of work on algorithms for refuting random CSPs, and 3SAT in particular; the listed papers contain additional references.} \cite{DBLP:conf/focs/AllenOW15,DBLP:journals/corr/RaghavendraRS16}, lower bounds: \cite{DBLP:journals/tcs/Grigoriev01,DBLP:conf/focs/Schoenebeck08,DBLP:journals/corr/KothariMOW17}), and tensor principal components analysis (upper bound: \cite{DBLP:conf/colt/HopkinsSS15,DBLP:journals/corr/RaghavendraRS16,DBLP:journals/eccc/BhattiproluGGLT16}, lower bound: this paper).

We also remark that our work applies to several previously-considered distinguishing and average-case problems within the sum-of-squares algorithmic framework: block models \cite{DBLP:conf/stoc/MontanariS16} , densest-$k$-subgraph \cite{MR2743268-Bhaskara10}; for each of these problems, we have by \pref{thm:maindist} an equivalence between efficient sum-of-squares algorithms and efficient spectral algorithms, and it remains to establish exactly what the tradeoff is between efficiency of the algorithm and the difficulty of distinguishing, or the strength of the noise.

To the best of knowledge, no previous work has attempted to characterize \sos relaxations for planted problems by simpler algorithms in the generality we do here.
Some works have considered characterizing degree-$2$ \sos relaxations (i.e. basic semidefinie programs) in terms of simpler algorithms.
One such example is recent work of Fan and Montanari \cite{DBLP:journals/corr/FanM16} who showed that for some planted problems on sparse random graphs, a class of simple procedures called \emph{local algorithms} performs as well as semidefinite programming relaxations.

\paragraph{On strong \sos lower bounds for planted problems}
By now, there's a large body of work that establishes lower bounds on \sos SDP for various average case problems. Beginning with the work of Grigoriev \cite{DBLP:journals/cc/Grigoriev01}, a long line work have established tight lower bounds for random constraint satisfaction problems \cite{DBLP:conf/focs/Schoenebeck08,MR3388187-Barak15,DBLP:journals/corr/KothariMOW17} and planted clique \cite{MR3388186-Meka15, DBLP:conf/colt/DeshpandeM15, DBLP:journals/corr/HopkinsKP15,DBLP:journals/corr/RaghavendraS15,DBLP:conf/focs/BarakHKKMP16}.  The recent \sos lower bound for planted clique of \cite{DBLP:conf/focs/BarakHKKMP16} was particularly influential to this work, setting the stage for our main line of inquiry.
We also draw attention to previous work on lower bounds for the tensor PCA and sparse PCA problems in the degree-$4$  \sos relaxation \cite{DBLP:conf/colt/HopkinsSS15,DBLP:journals/corr/MaW15}---our paper improves on this and extends our understanding of lower bounds for tensor and sparse PCA to any degree.

\Pnote{}
\Tnote{}
Tensor principle component analysis was introduced by Montanari and Richard \cite{DBLP:conf/nips/RichardM14} who indentified information theoretic threshold for recovery of the planted component and analyzed the maximum likelihood estimator for the problem. The work of \cite{DBLP:conf/colt/HopkinsSS15} began the effort to analyze the sum of squares method for the problem and showed that it yields an efficient algorithm for recovering the planted component with strength $\tilde{\omega}(n^{3/4})$. They also established that this threshold is  tight for the sum of squares relaxation of degree 4. Following this, Hopkins et al. \cite{DBLP:conf/stoc/HopkinsSSS16} showed how to extract a linear time spectral algorithm from the above analysis.
Tomioka and Suzuki derived tight information theoretic thresholds for detecting planted components by establishing tight bounds on the injective tensor norm of random tensors \cite{tomioka2014}.
Finally, very recently, Raghavendra et. al. and Bhattiprolu et. al. independently showed sub-exponential time algorithms for tensor pca \cite{DBLP:journals/corr/RaghavendraRS16,DBLP:journals/corr/BhattiproluGL16}. Their algorithms are spectral and are captured by the sum of squares method.

\Snote{}

\subsection{Organization}
\Snote{}
In \pref{sec:low-deg-dist} we set up and state our main theorem on \sos algorithms versus low-degree spectral algorithms.
In \pref{sec:examp} we show that the main theorem applies to numerous planted problems---we emphasize that checking each problem is very simple (and barely requires more than a careful definition of the planted and uniform distributions).
In \pref{sec:moment-match} and \pref{sec:proofofthm} we prove the main theorerm on \sos algorithms versus low-degree spectral algorithms.

In section 7 we get prepared to prove our lower bound for tensor PCA by proving a structural theorem on factorizations of low-degree matrix polynomials with well-behaved Fourier transforms.
In section 8 we prove our lower bound for tensor PCA, using some tools proved in section 9.

\paragraph{Notation}
For two matrices $A,B$, let $\iprod{A,B} \defeq \Tr(A B)$.
Let $\norm{A}_{Fr}$ denote the Frobenius norm, and $\norm{A}$ its spectral norm.
For matrix valued functions $A, B$ over $\scrI$ and a distribution $\udist$ over $\cI \sim \scrI$, we will denote $\iprod{A,B}_{\udist} = \E_{\cI \sim \udist} \iprod{A(\cI),B(\cI)}$ and by $\norm{A}_{Fr,\udist} \defeq \Paren{\E_{\cI \sim \udist} \iprod{A(\cI),A(\cI)}}^{1/2}$.

For a vector of formal variables $x = (x_1,\ldots,x_n)$, we use $x^{\leq d}$ to denote the vector consisting of all monomials of degree at most $d$ in these variables.
Furthermore, let us denote $X^{\leq d} \defeq (x^{\leq d})(x^{\leq d})^T$.

\section{Distinguishing Problems and Robust Inference}\label{sec:low-deg-dist}

In this section, we set up the formal framework within which we will prove our main result.

\subsubsection*{Uniform vs. Planted Distinguishing Problems}

We begin by describing a class of {\it distinguishing} problems.
For $\cA$ a set of real numbers, we will use $\scrI = \cA^N$ denote a space of instances indexed by $N$ variables---for the sake of concreteness, it will be useful to think of $\scrI$ as $\bits^{\instN}$; for example, we could have $N = \binom{n}{2}$ and $\scrI$ as the set of all graphs on $n$ vertices.
However, the results that we will show here continue to hold in other contexts, where the space of all instances is $\R^{\instN}$ or $[q]^{\instN}$.

\begin{definition}[Uniform Distinguishing Problem]
    Suppose that $\scrI$ is the space of all instances, and suppose we have two distributions over $\scrI$, a product distribution $\udist$ (the ``uniform'' distribution), and an arbitrary distribution $\pdist$ (the ``planted'' distribution).

    In a {\em uniform distinguishing problem}, we are given an instance $\cI \in \scrI$ which is sampled with probability $\frac{1}{2}$ from $\udist$ and with probability $\frac{1}{2}$ from $\pdist$, and the goal is to determine with probability greater than $\frac{1}{2} + \eps$ which distribution $\cI$ was sampled from, for any constant $\eps > 0$.
\end{definition}

\subsubsection*{Polynomial Systems}
In the uniform distinguishing problems that we are interested in, the planted distribution $\pdist$ will be a distribution over instances that obtain a large value for some optimization problem of interest (i.e. the max clique problem).
We define polynomial systems in order to formally capture optimization problems.

\begin{program}[Polynomial System]\label{prog:bopt}
    Let $\cA,\cB$ be sets of real numbers, let $n,N \in \N$, and let $\scrI = \cA^N$ be a space of instances and $\sols \subseteq \cB^n$ be a space of solutions.
    A {\em polynomial system} is a set of polynomial equalities
\begin{align*}
     &\quad g_j(x, \cI) = 0 \quad \forall j\in[m],
\end{align*}
    where $\{g_j\}_{j=1}^{m}$ are polynomials in the {\em program variables} $\{x_i\}_{i\in[n]}$, representing $x \in \sols$, and in the {\em instance variables} $\{\cI_j\}_{j\in[N]}$, representing $\cI \in \scrI$.
    We define $\deg_{\prog}(g_j)$ to be the degree of $g_j$ in the program variables, and $\deg_{\inst}(g_j)$ to be the degree of $g_j$ in the instance variables.
\end{program}

\begin{remark}
	For the sake of simplicity, the polynomial system \prettyref{prog:bopt} has no inequalities.
	Inequalities can be incorporated in to the program by converting each inequality in to an equality with an additional slack variable.
	Our main theorem still holds, but for some minor modifications of the proof, as outlined in \pref{sec:proofofthm}.
\end{remark}

A polynomial system allows us to capture problem-specific objective functions as well as problem-specific constraints.
For concreteness, consider a quadtratic program which checks if a graph on $n$ vertices contains a clique of size $k$.
We can express this with the polynomial system over program variables $x \in \R^n$ and instance variables $\cI \in \{0,1\}^{\binom{n}{2}}$, where $\cI_{ij} = 1$ iff there is an edge from $i$ to $j$, as follows:
\[
    \left\{{\small\sum}_{i\in[n]} x_i -k = 0\right\}\cup \{x_i(x_i - 1) = 0\}_{i\in[n]} \cup \{(1-\cI_{ij})x_i x_j = 0\}_{i,j \in \binom{[n]}{2}}.
    \]

\medskip

\subsubsection*{Planted Distributions}
We will be concerned with planted distributions of a particular form; first, we fix a polynomial system of interest $\cS = \{g_j(x,\cI)\}_{j\in[m]}$ and some set $\cX \subseteq \cB^n$ of feasible solutions for $\cS$, so that  the program variables $x$ represent elements of $\cX$.
Again, for concreteness, if $\scrI$ is the set of graphs on $n$ vertices, we can take $\sols \subseteq \bits^n$ to be the set of indicators for subsets of at least $n^{\eps}$ vertices.

For each fixed $\sol \in \sols$, let $\pdist_{|\sol}$ denote the uniform distribution over $\cI \in \scrI$ for which the polynomial system $\{g_j(x,\cI)\}_{j\in[m]}$ is feasible.
The planted distribution $\pdist$ is given by taking the uniform mixture over the $\pdist_{|\sol}$, i.e., $\pdist \sim U_{x\sim\sols} [\pdist_{|\sol}]$.

\subsubsection*{\sos Relaxations}
If we have a polynomial system $\{g_j\}_{j\in[m]}$ where $\deg_{\prog}(g_j) \le 2d$ for every $j \in [m]$, then the degree-$2d$ sum-of-squares SDP relaxation for the polynomial system \prettyref{prog:bopt} can be written as,

\begin{program}[\sos Relaxation for Polynomial System]\label{prog:boptmat}
    Let $\cS = \{g_j(x,\cI)\}_{j\in[m]}$ be a polynomial system in instance variables $\cI \in \scrI$ and program variables $x \in \sols$.
    If $\deg_{\prog}(g_j) \le 2d$ for all $j \in [m]$, then an {\em \sos relaxation} for $\cS$ is
\begin{align*}
	&\quad \iprod{G_j(\cI), X} = 0 \quad \forall j\in[m]\\
	& \quad X \succeq 0
\end{align*}
    where $X$ is an $[n]^{\le d} \times [n]^{\le d}$ matrix containing the variables of the SDP and  $G_j: \scrI \to \R^{[n]^{\leq d} \times [n]^{\leq d}} $ are matrices containing the coefficients of $g_j(x,\cI)$ in $x$, so that the constraint $\iprod{G_j(\cI),X} = 0$ encodes the constraint $g_j(x,\cI) = 0$ in the SDP variables.
    Note that the entries of $G_j$ are polynomials of degree at most $\deg_{\inst}(g_j)$ in the instance variables.
\end{program}

\subsubsection*{Sub-instances}
Suppose that $\scrI = \cA^N$ is a family of instances; then given an instance $\cI \in \scrI$ and a subset $S \subseteq [N]$, let $\cI_S$ denote the sub-instance consisting of coordinates within $S$.
Further, for a distribution $\subsetdist$ over subsets of $[N]$, let $\cI_{S} \sim_{\subsetdist} \cI$ denote a subinstance generated by sampling $S \sim \subsetdist$.
Let $\cI_{\downarrow}$ denote the set of all sub-instances of an instance $\cI$, and let $\scrI_{\downarrow}$ denote the set of all sub-instances of all instances.

\subsubsection*{Robust Inference}
Our result will pertain to polynomial systems that define planted distributions whose solutions to sub-instances generalize to feasible solutions over the entire instance.
We call this property ``robust inference.''

\begin{definition}
    Let $\scrI = \cA^N$ be a family of instances, let $\subsetdist$ be a distribution over subsets of $[N]$, let $\cS$ be a polynomial system as in \pref{prog:bopt}, and let $\mu$ be a planted distribution over instances feasible for $\cS$.
	Then the polynomial system $\cS$ is said to satisfy the {\em robust inference property for probability distribution $\pdist$ on $\scrI$ and subsampling distribution $\subsetdist$}, if given a subsampling $\cI_S$ of an instance $\cI$ from $\pdist$, one can infer a setting of the program variables $x^*$ that remains feasible to $\cS$ for most settings of $\cI_{\overline{S}}$.

Formally, there exists a map $ x : \scrI_{\subinst} \to \R^n$ such that
\[ \Pr_{\cI \sim \pdist, S \sim \subsetdist, \tilde{\cI} \sim
	\udist_{|\cI_S}}
		[x(\cI_{S}) \text{ is a feasible for
			$\cS$ on } \cI_S \circ \tilde{\cI}] \geq 1- \epsilon(n,d) \]
for some negligible function $\epsilon(n,d)$.  To specify the error probability, we will say that polynomial system is {\it $\epsilon(n,d)$-robustly inferable}.
\end{definition}

\subsubsection*{Main Theorem}
We are now ready to state our main theorem.

\begin{theorem} \label{thm:main} \label{thm:low-deg}
    Suppose that $\cS$ is a polynomial system as defined in \prettyref{prog:bopt}, of degree at most $2d$ in the program variables and degree at most $k$ in the instance variables.
Let $B > d\cdot k \in \N$ such that
\begin{enumerate}
	\item The polynpomial system $\cS$ is $\frac{1}{n^{8B}}$-robustly inferable with respect to the planted distribution $\pdist$ and the sub-sampling distribution $\subsetdist$.

	\item For $\cI \sim \udist$, the polynomial system $\cS$ admits a degree-$d$ \sos refutation with numbers bounded by $n^B$ with probability at least $1 - \frac{1}{n^{8B}}$.

\end{enumerate}
Let $D \in \N$ be such that for any subset $\alpha \subseteq [N]$ with $|\alpha | \geq D - 2dk$,
\[
    \Pr_{S \sim \subsetdist} [ \alpha \subseteq S ] \leq \frac{1}{n^{8B}}
\]

There exists a degree $2D$ matrix polynomial $Q : \scrI \to \R^{[n]^{\leq d} \times [n]^{\leq d}}$ such that,
    \[ \frac{\E_{\cI \sim \pdist} [ \lambda_{max}^+ (Q(\cI))]}{ \E_{\cI \sim \udist} [ \lambda_{max}^+ (Q(\cI))]}  \geq n^{B/2} \]
\end{theorem}

\begin{remark}
	Our argument implies a stronger result that can be stated in terms of the eigenspaces of the subsampling operator.
	Specifically, suppose we define
	\[ \cS_{\eps} \defeq \left\{ \alpha ~|~ \Pr_{S \sim \subsetdist} \{ \alpha \subseteq S\} \leq \epsilon \right\} \]
	Then, the distinguishing polynomial exhibited by \prettyref{thm:main} satisfies $Q \in \Span\{ \text{ monomials } \cI_{\alpha} | \alpha \in \cS_{\eps} \}$.
	This refinement can yield tighter bounds in cases where all monomials of a certain degree are not equivalent to each other.
	For example, in the \plantedclique\ problem, each monomial consists of a subgraph and the right measure of the degree of a sub-graph is the number of vertices in it, as opposed to the number of edges in it.
\end{remark}

In \pref{sec:examp}, we will make the routine verifications that the conditions of this theorem hold for a variety of distinguishing problems: planted clique (\pref{lem:pc-ex}), refuting random CSPs (\pref{lem:csp-ex}, stochastic block models (\pref{lem:sbm-ex}), densest-$k$-subgraph (\pref{lem:dks-ex}), tensor PCA (\pref{lem:tpca-ex}), and sparse PCA (\pref{lem:spca-ex}).
Now we will proceed to prove the theorem.
\section{Moment-Matching Pseudodistributions}\label{sec:moment-match}

We assume the setup from \pref{sec:low-deg-dist}: we have a family of instances $\scrI = \cA^N$, a polynomial system $\cS = \{g_j(x,\cI)\}_{j \in [m]}$ with a family of solutions $\sols = \cB^n$, a ``uniform'' distribution $\udist$ which is a product distribution over $\scrI$, and a ``planted'' distribution $\pdist$ over $\scrI$ defied by the polynomial system $\cS$ as described in \pref{sec:low-deg-dist}.

The contrapositive of \pref{thm:low-deg} is that if $\cS$ is robustly inferable with respect to $\mu$ and a distribution over sub-instances $\Theta$, and if there is no spectral algorithm for distinguishing $\pdist$ and $\udist$, then with high probability there is no degree-$d$ \sos refutation for the polynomial system $\cS$ (as defined in \pref{prog:boptmat}).
To prove the theorem, we will use duality to argue that if no spectral algorithm exists, then there must exist an object which is in some sense close to a feasible solution to the \sos SDP relaxation.

Since each $\cI$ in the support of $\mu$ is feasible for $\cS$ by definition, a natural starting point is the \sos SDP solution for instances $\cI \sim_{\pdist} \scrI$.
With this in mind, we let $\Lambda : \scrI \to (\R^{[n]^{\le d}\times [n]^{\le d}})_+$ be an arbitrary function from the support of $\mu$ over $\scrI$ to PSD matrices.
In other words, we take
\[
\Lambda(\cI) = \hat \pdist(\cI) \cdot M(\cI)
\]
where $\hat \mu$ is the relative density of $\mu$ with respect to $\nu$, so that $\hat\mu(\cI) = \mu(\cI)/\nu(\cI)$, and  $M$ is some matrix valued function such that $M(\cI) \succeq 0$ and $\norm{M(\cI)} \leq B$ for all $\cI \in \scrI$.
Our goal is to find a PSD matrix-valued function $P$ that matches the low-degree moments of $\Lambda$ in the variables $\cI$, while being supported over most of $\scrI$ (rather than just over the support of $\pdist$).

The function $P:\scrI \to (\R^{[n]^{\le d}\times [n]^{\le d}})_+$ is given by the following exponentially large convex program over matrix-valued functions,
\begin{program}[Pseudodistribution Program]\label{prog:distrib}
\begin{align}
	\min & \quad \norm{P}_{Fr,\udist}^2 \label{eq:obj}\\
    s.t.&\quad \iprod{Q,P}_{\udist} = \iprod{Q,\Lambda'}_{\udist} \quad \forall Q:\scrI \to \R^{[n]^{\le d}\times [n]^{\le d}},\ \deg_{\inst}(Q)\le D\label{eq:low-deg}\\
    &\quad P \succeq 0\nonumber\\
    &\quad \Lambda' = \Lambda + \eta\cdot \Id, \quad 2^{-2^{2^n}}> \eta > 0 \label{eq:lambda-perturb}
\end{align}
\end{program}

The constraint \pref{eq:low-deg} fixes $\E\Tr(P)$, and so the objective function \pref{eq:obj} can be viewied as minimizing $\E\Tr(P^2)$, a proxy for the collision probability of the distribution, which is a measure of entropy.

\begin{remark}
    We have perturbed $\Lambda$ in \pref{eq:lambda-perturb} so that we can easily show that strong duality holds in the proof of \pref{claim:dual}.
    For the remainder of the paper we ignore this perturbation, as we can accumulate the resulting error terms and set $\eta$ to be small enough so that they can be neglected.
\end{remark}

The dual of the above program will allow us to relate the existence of an \sos refutation to the existence of a spectral algorithm.
\begin{program}[Low-Degree Distinguisher]\label{prog:disting}
    \begin{align*}
	\max&\quad \iprod{\Lambda, Q}_{\udist}\\
	s.t.&\quad Q:\scrI \to \R^{[n]^{\le d}\times [n]^{\le d}}, \ \deg_{\inst}(Q) \le D\\
	&\quad \|Q_+\|_{Fr,\udist}^2 \le 1,
    \end{align*}
    where $Q_+$ is the projection of $Q$ to the PSD cone.
\end{program}
\begin{claim}\label{claim:dual}
    \pref{prog:disting} is a manipulation of the dual of \pref{prog:distrib}, so that if \pref{prog:distrib} has optimum $c > 1$, \pref{prog:disting} as optimum at least $\Omega(\sqrt{c})$.
\end{claim}

Before we present the proof of the claim, we summarize its central consequence in the following theorem: if \pref{prog:distrib} has a large objective value (and therefore does not provide a feasible \sos solution), then there is a spectral algorithm.
\begin{theorem} \label{thm:duality}
	Fix a function $M : \scrI \to  \R^{[n]^{\le d}\times [n]^{\le d}}_+$ be such that $\Id \succeq M \succeq 0$.
	Let $\lambda^{+}_{\max}(\cdot)$ be the function that gives the largest non-negative eigenvalue of a matrix.
	Suppose $\Lambda = \pdist \cdot M$ then the optimum of \pref{prog:distrib} is equal to $\opt > 1$ only if there exists a low-degree matrix polynomial $Q$ such that,
	\[  \E_{\cI \sim \pdist} [ \lambda_{max}^+ (Q(\cI))  ] \geq \Omega( \sqrt{\opt}/n^d ) \]
while,
\[  \E_{\cI \sim \udist} [ \lambda_{max}^+(Q(\cI))  ] \leq  1
\mper\]
\end{theorem}
\begin{proof}
    By \pref{claim:dual}, if the value of \pref{prog:distrib} is $\opt > 1$, then there is a polynomial $Q$ achieves a value of $\Omega(\sqrt{\opt})$ for the dual.
    It follows that
	\[
	    \E_{\cI \sim \pdist} [ \lambda_{max}^+(Q(\cI))]  \geq \frac{1}{n^d} \E_{\cI \sim \pdist} [ \iprod{\Id, Q(\cI))}] \geq \frac{1}{n^d} \iprod{\Lambda, Q}_\nu = \Omega(\sqrt{\opt}/n^d),
	\]
	while
	\[
	    \E_{\cI \sim \udist} [\lambda_{max}^+(Q(\cI))]
	    \leq \sqrt{\E_{\cI \sim \udist}[\lambda_{\max}^+(Q(\cI))^2]}
	    \leq \sqrt{\E_{\cI\sim \nu}\norm{Q_+(\cI)}_{Fr}^2} \leq 1.
	    \]
\end{proof}
It is interesting to note that the specific structure of the PSD matrix valued function $M$ plays no role in the above argument---since $M$ serves as a proxy for monomials in the solution as represented by the program variables $x^{\otimes d}$, it follows that the choice of how to represent the planted solution is not critical.
Although seemingly counterintuitive, this is natural because the property of being distinguishable by low-degre distinguishers or by SoS SDP relaxations is a property of $\udist$ and $\pdist$.

We wrap up the section by presenting a proof of the \pref{claim:dual}.
\Tnote{}

\begin{proof}[Proof of \pref{claim:dual}]
    We take the Lagrangian dual of \pref{prog:distrib}.
    Our dual variables will be some combination of low-degree matrix polynomials, $Q$, and a PSD matrix $A$:
    \[
	\cL(P,Q,A)
	= \|P\|_{Fr,\udist}^2 - \iprod{Q, P - \Lambda'}_\nu - \iprod{A,P}_\udist \qquad s.t.\quad A \sge 0.
    \]
    It is easy to verify that if $P$ is not PSD, then $A$ can be chosen so that the value of $\cL$ is $\infty$.
    Similarly if there exists a low-degree polynomial upon which $P$ and $\Lambda$ differ in expectation, $Q$ can be chosen as a multiple of that polynomial so that the value of $\cL$ is $\infty$.

    Now, we argue that Slater's conditions are met for \pref{prog:distrib}, as $P = \Lambda'$ is strictly feasible.
    Thus strong duality holds, and therefore
    \[
	\min_P \max_{A\sge 0,Q} \cL(P,Q,A)  \le \max_{A\sge 0, Q} \min_P \cL(P,Q,A).
    \]
    Taking the partial derivative of $\cL(P,Q,A)$ with respect to $P$, we have
    \begin{align*}
	\frac{\partial}{\partial P} \cL(P,Q,A) &= 2\cdot P - Q - A.
    \end{align*}
    where the first derivative is in the space of functions from $\scrI \to \R^{[n]^{\le d}\times [n]^{\le d}}$.
    By the convexity of $\cL$ as a function of $P$, it follows that if we set $\tfrac{\partial}{\partial P} \cL = 0$, we will have the minimizer.
    Substituting, it follows that
    \begin{align}
	\min_P \max_{A\sge 0,Q} \cL(P,Q,A)
	&\le \max_{A\sge 0, Q} \frac{1}{4}\|A + Q\|_{Fr,\udist}^2 - \frac{1}{2}\iprod{Q, A + Q - \Lambda'}_\nu - \frac{1}{2}\iprod{A,A+Q}_\nu\nonumber \\
	&= \max_{A\sge 0, Q} \iprod{Q,\Lambda'}_\nu - \frac{1}{4}\|A + Q\|_{Fr,\udist}^2\label{eq:optbd}
    \end{align}
    Now it is clear that the maximizing choice of $A$ is to set $A = -Q_{-}$, the negation of the negative-semi-definite projection of $Q$.
    Thus \pref{eq:optbd} simplifies to
    \begin{align}
	\min_P \max_{A\sge 0,Q} \cL(P,Q,A)
	&\le \max_{Q} \iprod{Q,\Lambda'}_\nu - \frac{1}{4}\|Q_+\|_{Fr,\udist}^2\nonumber\nonumber \\
	&\le \max_{Q} \iprod{Q,\Lambda}_\nu  + \eta \Tr_\nu(Q_+) - \frac{1}{4}\|Q_+\|_{Fr,\udist}^2,\label{eq:unconst}
    \end{align}
    where we have used the shorthand $\Tr_\nu(Q_+) \defeq \E_{\cI\sim \nu} \Tr(Q(\cI)_+)$.
    Now suppose that the low-degree matrix polynomial $Q^*$ achieves a right-hand-side value of
    \[
	\iprod{Q^*, \Lambda}_\nu + \eta\cdot \Tr_\nu(Q_+^*) - \frac{1}{4}\|Q_+^*\|_{Fr,\udist}^2 \ge c.
    \]
Consider $Q' = Q^*/\|Q^*_+\|_{Fr,\udist}$.
    Clearly $\|Q'_+\|_{Fr,\udist} = 1$.
    Now, multiplying the above inequality through by the scalar $1/\|Q^*_+\|_{Fr,\udist}$, we have that
    \begin{align*}
	\iprod{Q', \Lambda}_\nu
	&\ge \frac{c}{\|Q^*_+\|_{Fr,\udist}} - \eta \cdot \frac{\Tr_\nu(Q_+^*)}{\|Q_+^*\|_{Fr,\udist}}+ \frac{1}{4}\|Q^*_+\|_{Fr,\udist} \\
	&\ge \frac{c}{\|Q^*_+\|_{Fr,\udist}} - \eta \cdot n^d + \frac{1}{4}\|Q^*_+\|_{Fr,\udist}.
    \end{align*}
    Therefore $\iprod{Q',\Lambda}_\nu$ is at least $\Omega(c^{1/2})$, as if $\|Q^*_+\|_{Fr,\udist} \ge \sqrt{c}$ then the third term gives the lower bound, and otherwise the first term gives the lower bound.

    Thus by substituting $Q'$, the square root of the maximum of \pref{eq:unconst} within an additive $\eta n^d$ lower-bounds the maximum of the program
    \begin{align*}
	\max & \qquad \iprod{Q,\Lambda}_\nu\\
	s.t. & \qquad Q: \scrI \to \R^{[n]^{\le d}\times [n]^{\le d}},\quad \deg_{\inst}(Q) \le D\\
	&\qquad \|Q_+\|_{Fr,\udist}^2 \le 1.
    \end{align*}
    This concludes the proof.
\end{proof}
\section{Proof of \pref{thm:main}} \label{sec:proofofthm}

We will prove \pref{thm:main} by contradiction.  Let us assume that there exists no degree-$2D$ matrix polynomial that distinguishes $\udist$ from $\pdist$.
First, the lack of distinguishers implies the following fact about scalar polynomials.
    \begin{lemma} \label{lem:scalar}
	Under the assumption that there are no degree-$2D$ distinguishers, for every degree-$D$ scalar polynomial $Q$,
		\[ \norm{Q}^2_{Fr,\pdist} \leq n^{B} \norm{Q}^2_{Fr,\udist} \]
	\end{lemma}
	\begin{proof}
	    Suppose not, then the degree-$2D$ $1 \times 1$ matrix polynomial $\Tr(Q(\cI)^2)$ will be a distinguisher between $\pdist$ and $\udist$.
	\end{proof}

\paragraph{Constructing $\Lambda$}
First, we will use the robust inference property of $\pdist$ to construct a pseudo-distribution $\Lambda$.
Recall again that we have defined $\hat \mu$ to be the relative density of $\mu$ with respect to $\nu$, so that $\hat\mu(\cI) = \mu(\cI)/\nu(\cI)$.
For each subset $S \subseteq [N]$, define a PSD matrix-valued function $\Lambda_S : \scrI \to (\R^{[n]^{\leq d} \times [n]^{\leq d}})_+$ as,
\[ \Lambda_S( \cI ) = \E_{\cI'_{\overline{S}}} [\hat\mu(\cI_S\circ \cI'_{\overline{S}})] \cdot x(\cI_S)^{\leq d} (x(\cI_S)^{\leq d})^T    \]
	where we use $\cI_S$ to denote the restriction of $\cI$ to $S \subset [N]$, and $\cI_S \circ \cI'_{\overline{S}}$ to denote the instance given by completing the sub-instance $\cI_S$ with the setting $\cI'_{\overline{S}}$.
Notice that $\Lambda_S$ is a function depending only on $\cI_S$---this fact will be important to us.
Define $\Lambda \defeq \E_{S \sim \subsetdist} \Lambda_S$.
Observe that $\Lambda$ is a PSD matrix-valued function that satisfies
\begin{equation} \label{eq:lambdatopleft}
\iprod{\Lambda_{\emptyset,\emptyset}, 1}_{\udist}
    = \E_{\cI \sim \udist} \E_{S \sim \subsetdist} \E_{\cI'_{\overline S}\sim \udist} [\hat \mu(\cI_S \circ \cI'_{\overline S})]
    = \E_{S}\E_{\cI_{\overline S}} \E_{\cI_S \circ \cI'_{\overline{S}}\sim \nu}[\hat \mu (\cI_S \circ \cI'_{\overline{S}}) ]
    = 1
\end{equation}
Since $\Lambda(\cI)$ is an average over $\Lambda_S(\cI)$, each of which is a feasible solution with high probability, $\Lambda(\cI)$ is close to a feasible solution to the SDP relaxation for $\cI$.  The following Lemma formalizes this intuition.

Define $\cG \defeq \Span\{ \chi_S \cdot G_j ~|~ j \in [m], S \subseteq [N]\}$, and use $\Pi_{\cG}$ to denote the orthogonal projection into $\cG$.
\begin{lemma} \label{lem:lambdaideal}
	Suppose \prettyref{prog:bopt} satisfies the $\epsilon$-robust inference property with respect to planted distribution $\pdist$ and subsampling distribution $\subsetdist$ and if $\norm{x(\cI_S)^{\leq d}}_2^2 \leq K$ for all $\cI_S$ then
	for every $G \in \cG$, we have
	\[
	\iprod{\Lambda, G }_{\udist} \leq \sqrt{\epsilon } \cdot K \cdot \left( \E_{S \sim \subsetdist} \E_{\tilde{\cI}_{\overline{S}} \sim \udist} \E_{\cI \sim \pdist} \norm{G(\cI_S\circ\cI_{\overline{S}})}_2^2 \right)^{\nfrac{1}{2}}
	\]
\end{lemma}

\begin{proof}
    We begin by expanding the left-hand side by substituting the definition of $\Lambda$.
    We have
\begin{align*}
	\iprod{ \Lambda, G}_{\udist}  & = \E_{S \sim \subsetdist} \E_{\cI \sim \udist} \iprod{ \Lambda_S(\cI_S) , G(\cI)}\\
	& =  \E_{S \sim \subsetdist} \E_{\cI \sim \udist} \E_{\cI'_{\overline{S}}\sim \udist} \hat\pdist(\cI_S\circ \cI'_{\overline{S}}) \cdot \iprod{x(\cI_S)^{\leq d} (x(\cI_S)^{\leq d})^T , G(\cI)} \\
	\intertext{And because the inner product is zero if $x(\cI_S)$ is a feasible solution,}
	& \leq   \E_{S \sim \subsetdist} \E_{\cI \sim \udist} \E_{\cI'_{\overline{S}}\sim \udist} \hat \pdist(\cI_S\circ \cI'_{\overline{S}}) \cdot \Ind[ x(\cI_S) \text{ is infeasible for } \cS(\cI)]\cdot \Norm{x(\cI_S)^{\leq d} }^2_2 \cdot \norm{G(\cI)}_{Fr} \\
	& \leq  \E_{S \sim \subsetdist} \E_{\cI \sim \udist} \E_{\cI'_{\overline{S}}\sim \nu} \hat \pdist(\cI_S\circ\cI'_{\overline{S}}) \cdot \Ind[ x(\cI_S) \text{ is infeasible for } \cS(\cI)] \cdot K \cdot \norm{G(\cI)}_{Fr} \\
	\intertext{And now letting $\tilde{\cI}_{\overline{S}}$ denote the completion of $\cI_S$ to $\cI$, so that $\cI_S \circ \tilde{\cI}_{\overline{S}} = \cI$, we note that the above is like sampling $\cI'_{\overline{S}},\tilde{\cI}_{\overline{S}}$ independently from $\udist$ and then reweighting by $\hat\mu(\cI_S \circ \cI'_{\overline{S}})$, or equivalently taking the expectation over $\cI_S \circ \cI'_{\overline{S}} =\cI' \sim \pdist$ and $\tilde{\cI}_{\overline{S}} \sim \udist$:}
	& = \E_{S \sim \subsetdist} \E_{\cI' \sim \pdist} \E_{\tilde{\cI}_{\overline{S}} \sim \udist}  \cdot \Ind[ x(\cI_S) \text{ is infeasible for } \cS(\cI_S\circ \tilde{\cI}_{\overline{S}})] \cdot K \cdot \norm{G(\cI_S\circ \tilde{\cI_S})}_{Fr} \\
	\intertext{and by Cauchy-Schwarz,}
	& \leq K \cdot \left(\E_{S \sim \subsetdist} \E_{\cI' \sim \pdist} \E_{\tilde{\cI}_{\overline{S}} \sim \udist}  \cdot \Ind[ x(\cI_S) \text{ is infeasible for } \cS(\cI_S\circ \tilde{\cI}_{\overline{S}})]  \right)^{\nfrac{1}{2}}\cdot  \left(\E_{S \sim \subsetdist} \E_{\cI' \sim \pdist} \E_{\tilde{\cI}_{\overline{S}} \sim \udist}  \norm{G(\cI_S\circ \tilde{\cI_S})}^2_{Fr} \right)^{\nfrac{1}{2}}\\
\end{align*}
The lemma follows by observing that the first term in the product above is exactly the non-robustness of inference probability $\epsilon$.
\end{proof}

\begin{corollary} \label{cor:lambdaideal}
	If $G \in \cG$ is a degree-$D$ polynomial in $\cI$, then under the assumption that there are no degree-$2D$ distinguishers for $\udist,\pdist$,
	\[  \iprod{\Lambda , G}_{\udist} \leq \sqrt{\epsilon } \cdot K \cdot n^{B} \cdot \norm{G}_{Fr,\udist}  \]
\end{corollary}
\begin{proof}
	For each fixing of $\tilde{\cI_S}$, $\norm{G(\cI_S\circ \tilde{\cI_S})}_2^2$ is a degree-$2D$-scalar polynomial in $\cI$.
	Therefore by \pref{lem:scalar} we have that,
		\[ \E_{\cI \sim \pdist} \norm{G(\cI_S\circ \tilde{\cI_S})}_{Fr}^2 \leq n^B \cdot \E_{\cI \sim \udist} \norm{G(\cI_S\circ \tilde{\cI_S})}_{Fr}^2 \mper \]
	Substituting back in the bound in \pref{lem:lambdaideal} the corollary follows.
\end{proof}

Now, since there are no degree-$D$ matrix distinguishers $Q$, for each $S$ in the support of $\Theta$ we can apply reasoning similar to \pref{thm:duality} to conclude that there is a high-entropy PSD matrix-valued function $P_S$ that matches the degree-$D$ moments of $\Lambda_S$.

\begin{lemma}
    If there are no degree-$D$ matrix distinguishers $Q$ for $\pdist,\udist$, then for each $S \sim \Theta$, there exists a solution $P_S$ to \pref{prog:distrib} (with the variable $\Lambda := \Lambda_S$) and
\begin{equation} \label{eq:boundonP}
	\norm{P_S}_{Fr, \udist} \leq n^{\nfrac{(B + d)}{4}} \leq n^{B/2}
\end{equation}
\end{lemma}
This does not follow directly from \pref{thm:duality}, because a priori a distinguisher for some specific $S$ may only apply to a small fraction of the support of $\mu$.
However, we can show that \pref{prog:distrib} has large value for $\Lambda_S$ only if there is a distinguisher for $\mu,\nu$.
\begin{proof}
    By \pref{claim:dual}, it suffices for us to argue that there is no degree-$D$ matrix polynomial $Q$ which has large inner product with $\Lambda_S$ relative to its Frobenius norm.
    So, suppose by way of contradiction that $Q$ is a degree-$D$ matrix that distinguishes $\Lambda_S$, so that $\iprod{Q,\Lambda_S}_\nu \ge n^{B+d}$ but $\|Q\|_{Fr,\nu} \le 1$.

    It follows by definition of $\Lambda_S$ that
    \begin{align*}
	n^{B+d} \le \iprod{Q,\Lambda_S}_\nu
	&= \E_{\cI \sim \nu}\E_{\cI'_{\overline{S}}\sim \nu} \hat\mu(\cI_S \circ\cI'_{\overline S})\cdot \iprod{Q(\cI),x(\cI_S)^{\le d}(x(\cI_S)^{\le d})^{\top}}\\
	&= \E_{\cI_S \circ \cI'_{\overline S} \sim \mu}\Iprod{\E_{\cI_{\overline{S}}\sim \nu}Q(\cI_S\circ\cI_{\overline{S}}),x(\cI_S)^{\le d}(x(\cI_S)^{\le d})^{\top}}\\
	&\le \E_{\mu} \left[\lambda^{+}_{\max}\Paren{\E_{\cI_{\overline{S}}\sim \nu}Q(\cI_S\circ\cI_{\overline{S}})}\right] \cdot \Norm{x(\cI_S)^{\le d}}^2_2\mper
    \end{align*}
    So, we will show that $Q_S(\cI) = \E_{\cI'_{\overline S}\sim \nu} Q(\cI_{S}\circ \cI'_{\overline{S}})$ is a degree-$D$ distinguisher for $\mu$.
    The degree of $Q_S$ is at most $D$, since averaging over settings of the variables cannot increase the degree.
    Applying our assumption that $\|x(\cI_S)^{\le d}\|_2^2 \le K\le n^d$, we already have $\E_\mu\lambda^+_{\max}(Q_S) > n^{B}$.
    It remains to show that $\E_\nu \lambda_{\max}^+ (Q_S)$ is bounded.
    For this, we use the following fact about the trace.
    \begin{fact}[See e.g. Theorem 2.10 in \cite{Carlen2009TraceIA}]
	For a function $f:\R \to \R$ and a symmetric matrix $A$ with eigendecomposition $\sum \lambda\cdot vv^\top$, define $f(A) = \sum f(\lambda)\cdot vv^\top$.
	If $f:\R \to \R$ is continuous and convex, then the map $A \to \Tr(f(A))$ is convex for symmetric $A$.
    \end{fact}
    The function $f(t) = \left(\max\{0,t\}\right)^2$ is continuous and convex over $\R$, so the fact above implies that the map $A \to \|A_+\|_{Fr}^2$ is convex for symmetric $A$.
    We can take $Q_S$ to be symmetric without loss of generality, as in the argument above we only consider the inner product of $Q_S$ with symmetric matrices.
    Now we have that
    \[
	\Norm{\Paren{Q_S(I)}_+}_{Fr}^2
	=   \Norm{\left(\E_{\cI'_{\overline S}}\left[Q(\cI_S \circ \cI'_{\overline S})\right]\right)_+}_{Fr}^2
	\le
	\E_{\cI'_{\overline S}}\Norm{\Paren{Q(\cI_S \circ \cI'_{\overline S})}_+}_{Fr}^2,
    \]
where the inequality is the definition of convexity.
    Taking the expectation over $\cI\sim \nu$ gives us that $\|(Q_S)_+\|_{Fr,\nu}^2 \le \|Q_+\|_{Fr,\nu}^2 \le 1$, which gives us our contradiciton.
\end{proof}
Now, analogous to $\Lambda$, set $P \defeq \E_{S \sim \subsetdist} P_S$.

\paragraph{Random Restriction}
We will exploit the crucial property that $\Lambda$ and $P$ are averages over functions that depend on subsets of variables.
This has the same effect as a random restriction, in that $\iprod{P,R}_\nu$ essentially depends on the low-degree part of $R$.
Formally, we will show the following lemma.

\begin{lemma} \label{lem:randomrestriction} (Random Restriction)
	Fix $D, \ell \in \N$.
	For matrix-valued functions $R : \scrI \to \R^{\ell \times \ell}$ and a family of functions $\{P_S : \scrI_S \to \R^{\ell \times \ell} \}_{S \subseteq [N]}$, and a distribution $\subsetdist$ over subsets of $[N]$,
	\[
	\E_{\cI \sim \udist} \E_{S \sim \subsetdist} \iprod{ P_S(\cI_S), R(\cI)} \geq  \E_{S \sim \subsetdist} \E_{\cI \sim \udist}
	\iprod{P_S(\cI_S), R^{<D}_S(\cI_S)} - \rho(D,\subsetdist)^{\nfrac{1}{2}} \cdot \left(\E_{S \sim \subsetdist} \norm{P_S}_{Fr, \udist}^2 \right)^{\frac{1}{2}}
	\norm{R}_{Fr,\udist}
	\]
	where
	\[
	\rho(D,\subsetdist)  = \max_{\alpha, |\alpha| \geq D} \Pr_{S \sim \subsetdist} [ \alpha \subseteq S].
	\]
\end{lemma}
\begin{proof}
    We first re-express the left-hand side as
\begin{align*}
	\E_{\cI \sim \udist} \E_{S \sim \subsetdist} \iprod{ P_S(\cI_S), R(\cI)} = \E_{S \sim \subsetdist} \E_{\cI \sim \udist} \iprod{P_S(\cI_S),
  R_S(\cI_S)}
\end{align*}
where $R_S(\cI_S) \defeq \E_{ \cI_{\overline{S}}}[R(\cI)]$ obtained by averaging out all coordinates outside $S$.
    Splitting the function $R_S$ into its low-degree and high-degree parts, $R_S = R_S^{\leq D} + R_S^{> D}$, then applying a Cauchy-Schwartz inequality we get
\begin{align*}
	\E_{S \sim \subsetdist} \E_{\cI \sim \udist} \iprod{P_S(\cI_S), R_S(\cI_S)} & \geq \E_{S \sim \subsetdist} \E_{\cI \sim \udist}
	\iprod{P_S(\cI_S), R^{<D}_S(\cI_S)} - \left(\E_{S \sim \subsetdist} \norm{P_S}^2_{Fr, \udist} \right)^{\nfrac{1}{2}} \cdot
	\left( \E_{S \sim \subsetdist} \norm{R^{\geq D}_S}^2_{Fr,\udist} \right)^{\nfrac{1}{2}}\mper
\end{align*}
    Expressing $R^{\ge D}(\cI)$ in the Fourier basis, we have that over a random choice of $S\sim \Theta$,
$$  \E_{S \sim \subsetdist} \norm{R^{\geq D}_S}_{Fr,\udist}^2  = \sum_{\alpha,|\alpha| \geq D} \Pr_{S \sim \subsetdist} [\alpha \subseteq S] \cdot \hat{R}_\alpha^2
\leq \rho(D, \subsetdist)  \cdot \norm{R}_{Fr}^2 $$
Substituting into the above inequality, the conclusion follows.
\end{proof}
\paragraph{Equality Constraints}
Since $\Lambda$ is close to satisfying all the equality constraints $\cG$ of the SDP, the function $P$ approximately satisfies the low-degree part of $\cG$.
Specifically, we can prove the following.
\begin{lemma}\label{lem:ideal}
    Let $k\ge \deg_{\inst}(G_j)$ for all $G_j \in \cS$.
    With $P$ defined as above and under the conditions of \prettyref{thm:main} for any function $G \in \cG$,
	\[
		\left|\iprod{P, G^{\le D}}_{\udist}\right| \le \frac{2}{n^{2B}} \norm{G}_{Fr,\udist}  \]
\end{lemma}

\begin{proof}
	Recall that $\cG = \Span\{\chi_S \cdot G_j ~|~ j \in [m], S \subseteq [N]\}$
    and let $\Pi_{\cG}$ be the orthogonal projection into $\cG$.
    Now, since $G \in \cG$,
    \begin{align}
	G^{\le D}
	&= (\Pi_{\cG} G)^{\le D}
	= (\Pi_{\cG} G^{\le D - 2k})^{\le D} + (\Pi_{\cG} G^{> D - 2k})^{\le D}.\label{eq:sep}
    \end{align}

	Now we make the following claim regarding the effect of projection on to the ideal $\cG$, on the degree of a polynomial.
    \begin{claim}
	For every polynomial $Q$, $\deg(\Pi_{\cG} Q) \leq \deg(Q) + 2k$.
	Furthermore for all $\alpha$, $\Pi_{\cG} Q^{> \alpha}$ has no monomials of degree $\leq \alpha - k$
    \end{claim}
   \begin{proof}
	   To establish the first part of the claim it suffices to show that $\Pi_{\cG} Q \in \Span \{\chi_S\cdot G_j  ~|~ |S| \leq \deg(Q) + k \}$, since $\deg(G_j) \leq k$ for all $j \in [m]$.
	   To see this, observe that $\Pi_{\cG} Q \in \Span \{ \chi_S\cdot G_j  ~|~ |S| \leq \deg(Q) + k \}$ and is orthogonal to every $\chi_S\cdot G_j $ with $|S| > \deg(Q) + k$:
	\begin{align*}
	    \iprod{\Pi_{\cG} Q ,\chi_S \cdot G_j}_\nu
	    = \iprod{Q, \Pi_{\cG} \chi_S \cdot G_j}_\nu
= \iprod{Q, \chi_S \cdot G_j}_\nu
    = \iprod{Q  G_j, \chi_S}_\nu
	= 0,
    	\end{align*}
	where the final equality is because $\deg(\chi_S) > \deg(G_j) + \deg(Q)$.
	On the other hand, for every subset $S$ with $\deg(\chi_S) \leq \alpha - k$,
	\begin{align*}
		\iprod{ \Pi_{\cG} Q^{>\alpha} ,\chi_S \cdot G_j}
		= \iprod{Q^{> \alpha}, \Pi_{\cG} \chi_S \cdot G_j}
		= \iprod{Q^{> \alpha} , \chi_S \cdot G_j}
   	= 0, \text{ since } \alpha > \deg(G_j) + \deg(\chi_S)
    	\end{align*}
	This implies that $\Pi_{\cG} Q^{>\alpha} \in \Span\{ \chi_S \cdot G_j ~|~ |S| > \alpha -k \}$ which implies that $\Pi_{\cG} Q^{> \alpha}$ has no monomials of degree $\leq \alpha -k$.
   \end{proof}

    Incorporating the above claim into \pref{eq:sep}, we have that
    \[
	G^{\le D} = \Pi_{\cG} G^{\le D - 2k} + (\Pi_{\cG} G^{\ge D - 2k})^{[D-3k,D]},
    \]
where the superscript $[D-3k,D]$ denotes the degree range.
    Now,
    \begin{align*}
	    \iprod{P, G^{\le D}}_{\udist}
	    &= \iprod{P,  \Pi_{\cG} G^{\le D - 2k}}_{\udist} + \iprod{P,   (\Pi_{\cG} G^{\ge D - 2k})^{[D-3k,D]}}_{\udist}
	\intertext{And since $\Pi_{\cG} G^{\le D - 2k}$ is of degree at most $D$ we can replace $P$ by $\Lambda$,}
	&= \iprod{ \Lambda,  \Pi_{\cG} G^{\le D - 2k}}_{\udist} + \iprod{P,   (\Pi_{\cG} G^{\ge D - 2k})^{[D-3k,D]}}_{\udist} \\
	\intertext{Now bounding the first term using \pref{cor:lambdaideal} with a $n^B$ bound on $K$,}
	&\leq \left(\frac{1}{n^{8B}} \right)^{\nfrac{1}{2}} \cdot n^B \cdot ( n^B \cdot \norm{\Pi_{\cG} G_{\emptyset, \emptyset}^{\leq D - 2k}}_{Fr,\udist}) +\iprod{ P ,  (\Pi_{\cG} G^{\ge D - 2k})^{[D-3k,D]}}
	\intertext{And for the latter term we use \pref{lem:randomrestriction},}
	&\leq \frac{1}{n^{2B}} \norm{\Pi_{\cG} G^{\leq D - 2k}_{\emptyset,\emptyset}}_{Fr,\udist} + \frac{1}{n^{4B}} \left( \E_S \norm{P_S}_{Fr,\udist}^2 \right)^{\nfrac{1}{2}}   \norm{G}_{Fr,\udist},
    \end{align*}
    where we have used the fact that $(\Pi_{\cG} G^{\ge D - 2k})^{[D-3k,D]}$ is high degree.
    By property of orthogonal projections, $\|\Pi_{\cG} G^{\ge D-2k}\|_{Fr,\udist} \le \|G^{\ge D-2k}\|_{Fr,\udist} \le \|G\|_{Fr,\udist}$.  Along with the bound on $\norm{P_S}_{Fr,\udist}$ from \pref{eq:boundonP}, this implies the claim of the lemma.
\end{proof}
Finally, we have all the ingredients to complete the proof of \pref{thm:main}.
\begin{proof}[Proof of \pref{thm:main}]
Suppose we sample an instance $\cI \sim \udist$, and suppose by way of contradiction this implies that with high probability the \sos SDP relaxation is infeasible.
In particular, this implies that there is a degree-$d$ sum-of-squares refutation of the form,
    \[
	-1  = a^{\cI}(x)+ \sum_{j\in [m]} g_j^{\cI}(x) \cdot q^{\cI}_j(x) ,
    \]
where $a^{\cI}$ is a sum-of-squares of polynomials of degree at most $2d$ in $x$, and $\deg(q^{\cI}_j) + \deg(g^{\cI}_j) \le 2d$.
    Let $A^{\cI} \in \R^{[n]^{\le d} \times [n]^{\le d}}$ be the matrix of coefficients for $a^{\cI}(c)$ on input $\cI$, and let $G^{\cI}$ be defined similarly for $\sum_{j\in[m]} g_j(x)\cdot q_j(x)$.
    We can rewrite the sum-of-squares refutation as a matrix equality,
\[
	-1  = \iprod{ X^{\leq d}, A^{\cI} } + \iprod{ X^{\leq d}, G^{\cI}} ,
    \]
    where $G^{\cI} \in \cG$, the span of the equality constraints of the SDP.

 Define $\success : \scrI \to \bits$ as
    \[s(\cI) \defeq \Ind[ \exists \text{ a degree-$2d$ sos-refutation for } \cS(\cI)]\]
    By assumption, $\E_{ \cI \sim \udist} [s(\cI)] = 1 - \frac{1}{n^{8B}}$.
    Define matrix valued functions $A, G: \scrI \to \R^{[n]^{\leq d} \times [n]^{\leq d} }$ by setting,
    \[ A(\cI) \defeq s(\cI) \cdot A^{\cI} \]
 \[ G(\cI) \defeq s(\cI) \cdot G^{\cI} \]
With this notation, we can rewrite the sos-refutation identity as a polynomial identity in $X$ and $\cI$,
\[
	-s(\cI)  = \iprod{ X^{\leq d}, A(\cI) } + \iprod{ X^{\leq d}, G(\cI)} \mper
    \]
    Let $\mempty$ denote the $[n]^{\leq d} \times [n]^{\leq d}$ matrix with the entry corresponding to $(\emptyset, \emptyset)$ equal to $1$, while the remaining entries are zero.  We can rewrite the above equality as,
\[
	-\iprod{X^{\leq d}, s(\cI) \cdot \mempty}  = \iprod{ X^{\leq d}, A(\cI) } + \iprod{ X^{\leq d}, G(\cI)} \mper
    \]
for all $\cI$ and formal variables $X$.

Now, let $P = \E_{S \sim \subsetdist} P_S$ where each $P_S$ is obtained by from the \pref{prog:distrib} with $\Lambda_S$.
    Substituting $X^{\leq d}$ with $P(\cI)$ and taking an expectation over $\cI$,
\begin{align} \label{eq:contra}
	    \iprod{P, s(\cI) \cdot \mempty}_{\udist}
	    &=  \iprod{P,A}_{\udist}+   \iprod{P,G}_{\udist}  \\
	    &\ge  \iprod{P,G}_{\udist}
    \end{align}
    where the inequality follows because $A,P \sge 0$.  We will show that the above equation is a contradiction by proving that LHS is less than $-0.9$, while the right hand side is at least $-0.5$.
    First, the right hand side of \eqref{eq:contra} can be bounded by \pref{lem:ideal}
    \begin{align}
	    \iprod{P, G}_{\udist}
	    & = \E_{\cI \sim \udist} \E_{S \sim \subsetdist} \iprod{P_S(\cI_S), G(\cI)} \nonumber \\
	    & \geq \E_{\cI \sim \udist} \E_{S \sim \subsetdist} \iprod{P_S(\cI_S), G^{\leq D}(\cI)} - \frac{1}{n^{4B}} \cdot \left( \E_S \norm{P_S}_{Fr,\udist}^2 \right)^{1/2} \cdot \norm{G}_{Fr,\udist}   \qquad (\text{random restriction \prettyref{lem:randomrestriction}}) \nonumber \\
	    & \geq  - \frac{2}{n^{2B}} \cdot \norm{G}_{Fr,\udist} - \frac{1}{n^{4B}}
	    \left( \E_S \norm{P_S}_{Fr,\udist}^2 \right)^{\nfrac{1}{2}}
	    \norm{G}_{Fr,\udist} \qquad (\text{using
		    \pref{lem:ideal}}) \nonumber \\
	    & \geq -\frac{1}{2} \nonumber
     \end{align}
     where the last step used the bounds on $\norm{P_S}_{Fr,\udist}$ from \pref{eq:boundonP} and on $\norm{G}_{Fr,\udist}$ from the $n^{B}$ bound assumed on the \sos proofs in \pref{thm:main}.

     Now the negation of the left hand side of \eqref{eq:contra} is
    \begin{align}
	    \E_{\cI \sim \udist}  \iprod{P(\cI), s(\cI) \cdot \mempty }
	    &\ge \E_{\cI \sim \udist} [ P_{\emptyset, \emptyset}(\cI) \cdot 1 ] -  \E[(s - 1)^2]^{1/2} \cdot \|P\|_{Fr,\udist}\nonumber \\
	    \intertext{The latter term can be simplified by noticing that the expectation of the square of a 0,1 indicator is equal to the expectation of the indicator, which is in this case $\frac{1}{n^{8B}}$ by assumption.
	    Also, since $1$ is a constant, $P_{\emptyset,\emptyset}$ and $\Lambda_{\emptyset,\emptyset}$ are equivalent:}
	    &= \E_{\cI \sim \udist} [ \Lambda_{\emptyset, \emptyset}(\cI) \cdot  1 ] -  \frac{1}{n^{4B}} \cdot \|P\|_{Fr,\udist}\nonumber \\
	    &= 1  -  \frac{1}{n^{4B}} \cdot \|P\|_{Fr, \udist}   \qquad (\text{ using  \eqref{eq:lambdatopleft}} ) \nonumber \\
	    &= 1 - \frac{1}{n^{3B}}   \qquad (\text{using  \eqref{eq:boundonP}} ) \nonumber
    \end{align}
    We have the desired contradiction in \eqref{eq:contra}.
\end{proof}

\subsection{Handling Inequalities}

Suppose the polynomial system \pref{prog:bopt} includes inequalities of the form $h(\cI, x) \geq 0$, then a natural approach would be to introduce a slack variable $z$ and set $h(\cI,x) - z^2 = 0$.
Now, we can view the vector $(x,z)$ consisting of the original variables along with the slack variables as the hidden planted solution.
The proof of \pref{thm:main} can be carried out as described earlier in this section, with this setup.
However, in many cases of interest, the inclusion of slack variables invalidates the robust inference property.
This is because, although a feasible solution $x$ can be recovered from a subinstance $\cI_S$,  the value of the corresponding slack variables could potentially depend on $\cI_{\overline{S}}$.
For instance, in a random CSP, the value of the objective function on the assignment $x$ generated from $\cI_S$ depends on all the constraints outside of $S$ too.

The proof we described is to be modified as follows.
\begin{itemize}
	\item As earlier, construct $\Lambda_S$ using only the robust inference property of original variables $x$, and the corresponding matrix functions $P_S$.

	\item Convert each inequality of the form $h_i(\cI,x) \geq 0$, in to an equality by setting $h_i(\cI,x) = z_i^2$.

	\item Now we define a pseudo-distribution $\tilde{\Lambda}_S(\cI_S)$ over original variables $x$ and slack variables $z$ as follows.  It is convenient to describe the pseudo-distribution in terms of the corresponding pseudo-expectation operator.  Specifically, if $x(\cI_S)$ is a feasible solution for \pref{prog:bopt} then define
		\[ \tilde{E} [z_{\sigma} x_\alpha ]  \defeq \begin{cases} 0 & \text{ if } \sigma_i \text{ odd for some } i \\
				\prod_{i \in \sigma} (h_i(\cI,x(\cI_S)))^{\sigma_i/2} \cdot x(\cI_S)_\alpha & \text{ otherwise }  \end{cases} \]
		Intuitively, the pseudo-distribution picks the sign for each $z_i$ uniformly at random, independent of all other variables.  Therefore, all moments involving an odd power of $z_i$ are zero.  On the other hand, the moments of even powers of $z_i$ are picked so that the equalities $h_i(\cI,x) = z_i$ are satisfied.

		It is easy to check that $\tilde{\Lambda}$ is psd matrix valued, satisfies \pref{eq:lambdatopleft} and all the equalities.

	\item While $\Lambda_S$ in the original proof was a function of $\cI_S$,  $\tilde{\Lambda}_S$ is not.  However, the key observation is that, $\tilde{\Lambda}_S$ is degree at most $k\cdot d$ in the variables outside of $S$.  Each function $h_i(\cI, x(\cI_S))$ is degree at most $k$ in $\cI_{\overline{S}}$, and the entries of $\tilde{\Lambda}_S(\cI_S)$ are a product of at most $d$ of these polynomials.

	\item The main ingredient of the proof that is different from the case of equalities is the random restriction lemma which we outline below.
	    The error in the random restriction is multiplied by $D^{dk/2} \le n^{B/2}$; however this does not substantially change our results, since \pref{thm:main} requires $\rho(D,\Theta) < n^{-8B}$, which leaves us enough slack to absorb this factor (and in every application $\rho(D,\Theta) = p^{O(D)}$for some $p<1$ sufficiently small that we meet the requirement that $D^{dk}\rho(D-dk,\Theta)$ is monotone non-increasing in $D$).
\end{itemize}

\begin{lemma} \label{lem:mod-randomrestriction}[Random Restriction for Inequalities]
	Fix $D, \ell \in \N$.
	Consider a matrix-valued function $R : \scrI \to \R^{\ell \times \ell}$ and a family of functions $\{P_S : \scrI \to \R^{\ell \times \ell} \}_{S \subseteq [N]}$ such that each $P_S$ has degree at most $dk$ in $\cI_{\overline{S}}$.
	If $\subsetdist$ is a distribution over subsets of $[N]$ with
	\[ \rho(D,\subsetdist)  = \max_{\alpha, |\alpha| \geq D} \Pr_{S \sim \subsetdist} [ \alpha \subseteq S],\]
	and the additional requirement that $D^{dk}\cdot\rho(D-dk,\Theta)$ is monotone non-increasing in $D$, then
	$$\E_{\cI \sim \udist} \E_{S \sim \subsetdist} \iprod{ P_S(\cI_S), R(\cI)} \geq  \E_{S \sim \subsetdist} \E_{\cI \sim \udist}
	\iprod{P_S(\cI_S), \tilde R^{<D}_S(\cI_S)} - D^{dk/2} \cdot \rho(D-dk,\subsetdist)^{\nfrac{1}{2}} \cdot \left(\E_{S \sim \subsetdist} \norm{P_S}_{2, \udist}^2 \right)^{\frac{1}{2}}
	\norm{R}_{Fr,\udist} $$
\end{lemma}
\begin{proof}
\begin{align*}
	\E_{\cI \sim \udist} \E_{S \sim \subsetdist} \iprod{ P_S(\cI_S), R(\cI)} = \E_{S \sim \subsetdist} \E_{\cI \sim \udist} \iprod{P_S(\cI_S),
  \tilde R_S(\cI)}
\end{align*}
where $\tilde R_S(\cI)$ is now obtained by averaging out the values for all monomials whose degree in $\overline{S}$ is $> dk$. Writing $\tilde R_S = \tilde R_S^{\leq D}+\tilde R_S^{> D}$ and applying a Cauchy-Schwartz inequality we get,
\begin{align*}
    \E_{S \sim \subsetdist} \E_{\cI \sim \udist} \iprod{P_S(\cI_S), \tilde R_S(\cI)}
    & \geq \E_{S \sim \subsetdist} \E_{\cI \sim \udist}
	\iprod{P_S(\cI_S), \tilde R^{<D}_S(\cI)} - \left(\E_{S \sim \subsetdist} \norm{P_S}^2_{Fr, \udist} \right)^{\nfrac{1}{2}} \cdot
	\left( \E_{S \sim \subsetdist} \norm{\tilde R^{\geq D}_S}_{Fr,\udist} \right)^{\nfrac{1}{2}}
\end{align*}
Over a random choice of $S$,
$$  \E_{S \sim \subsetdist} \norm{\tilde R^{\geq D}_S}_{Fr,\udist}^2  = \sum_{\alpha,|\alpha| \geq D} \Pr_{S \sim \subsetdist} [| \alpha \cap \overline{S}| \leq dk ] \cdot \hat{R}_\alpha^2
\leq D^{dk} \cdot \rho(D-dk, \subsetdist)  \cdot \norm{R}_{Fr}^2, $$
    where we have used that $D^{dk} \rho(D-dk,\subsetdist)$ is a monotone non-increasing function of $D$.
Substituting this in the earlier inequality the Lemma follows.
\end{proof}
\section{Applications to Classical Distinguishing Problems} \label{sec:examp}
In this section, we verify that the conditions of \pref{thm:low-deg} hold for a variety of canonical distinguishing problems.
We'll rely upon the (simple) proofs in \pref{app:idcond}, which show that the ideal term of the \sos proof is well-conditioned.

\begin{problem}[Planted clique with clique of size $n^\delta$]
    Given a graph $G = (V,E)$ on $n$ vertices, determine whether it comes from:
    \begin{compactitem}
    \item {\bf Uniform Distribution}: the uniform distribution over graphs on $n$ vertices ($G(n,\tfrac{1}{2})$).
	\item {\bf Planted Distribution}: the uniform distribution over $n$-vertex graphs with a clique of size at least $n^{\delta}$
    \end{compactitem}
    The usual polynomial program for \emph{planted clique} in variables $x_1,\ldots, x_n$ is:
    \begin{align*}
	\obj &\le \sum_i x_i\\
	x_i^2 &= x_i\quad \forall i \in [n]\\
	x_ix_j &= 0\quad \forall (i,j) \in E
    \end{align*}
\end{problem}
\begin{lemma}\label{lem:pc-ex}
    \pref{thm:low-deg} applies to the above planted clique program, so long as  $\obj \le n^{\delta - \epsilon}$ for any $\epsilon \ge \frac{c\cdot d}{D-6d}$ for a fixed constant $c$.
\end{lemma}
\begin{proof}
    For planted clique, for our notion of ``instance degree'', rather than the multiplicity of instance variables, the ``degree'' of $\cI_{\alpha}$ will be the number of distinct vertices incident on the edges in $\alpha$.
    The proof of \pref{thm:main} proceeds identically with this notion of degree, but we will be able to achieve better bounds on $D$ relative to $d$.

    In this case, the instance degree of the \sos relaxation is $k = 2$.
    We have from \pref{cor:well-cond} that the degree-$d$ \sos refutation is well-conditioned, with numbers bounded by $n^{c_1 \cdot d}$ for some constant $c_1/2$.
    Define $B = c_1 d \ge dk$.

    Our subsampling distribution $\subsetdist$ is the distribution given by including every vertex with probability $\rho$, producing an induced subgraph of $\approx \rho n$ vertices.
    For any set of edges $\alpha$ of instance degree at most $D-6d$,
    \[
	\Pr_{S\sim \subsetdist}[\alpha \subseteq S] \le \rho^{D-6d},
    \]
    since the instance degree corresponds to the number of vertices incident on $\alpha$.

    This subsampling operation satisfies the subsample inference condition for the clique constraints with probability $1$, since a clique in any subgraph of $G$ is also a clique in $G$.
    Also, if there is a clique of size $n^{\delta}$ in $G$, then by a Chernoff bound
    \[
	\Pr_{S \sim \subsetdist}[\exists \text{ clique of size } \ge (1-\beta)\rho n^{\delta} \in S]
	\ge 1 - \exp(-\frac{\beta^2\rho n^{\delta}}{2})\mper
    \]
Choosing  $\beta = \sqrt{\frac{10 B \log n}{\rho n^{\delta}}}$, this gives us that $\Theta$ gives $n^{-10 B}$-robust inference for the planted clique problem, so long as $\obj \le \rho n/2$.
    Choosing $\rho = n^{-\epsilon}$ for $\epsilon$ so that
    \[
	\rho^{D-6d} \le n^{-8B}  \implies \eps \ge \frac{c_2 d}{D - 6d},
	\]
	for some constant $c_2$, all of the conditions required by \pref{thm:low-deg} now hold.
\end{proof}

\begin{problem}[Random CSP Refutation at clause density $\alpha$]
    Given an instance of a Boolean $k$-CSP with predicate $P:\{\pm 1\}^k \to \{\pm 1\}$ on $n$ variables with clause set $C$, determine whether it comes from:
    \begin{compactitem}
    \item {\bf Uniform Distribution}: $m \approx \alpha n$ constraints are generated as follows.
	Each $k$-tuple of variables $S \in [n]^k$ is independently with probability $p=\alpha n^{-k+1}$ given the constraint $P(x_S \circ z_S) = b_S$ (where $\circ$ is the entry-wise multiplication operation) for a uniformly random $z_S \in \{\pm 1 \}^k$ and $b_S \in \{\pm 1\}$.
    \item {\bf Planted Distribution}: a planted solution $y \in \{\pm 1\}^n$ is chosen, and then $m \approx \alpha n$ constraints are generated as follows.
	Each $k$-tuple of variables $S \in [n]^k$ is independently with probability $p = \alpha n^{-k+1}$ given the constraint $P(x_S \circ z_S) = b_S$ for a uniformly random $z_S \in \{\pm 1 \}^k$, but $b_S = P(y_S \circ z_S)$ with probability $1-\delta$ and $b_S$ is uniformly random otherwise.
    \end{compactitem}
    The usual polynomial program for \emph{random CSP refutation} in variables $x_1,\ldots, x_n$ is:
    \begin{align*}
	\obj &\le
	\sum_{S \in [n]^k} \Ind[\exists \text{ constraint on }S] \cdot \left(\frac{1+ P(x_S \circ z_S) \cdot b_S}{2}\right)\\
	x_i^2 &= 1 \quad \forall i \in [n]
    \end{align*}
\end{problem}
\begin{lemma}\label{lem:csp-ex}
    If $\alpha \ge 1$, then \pref{thm:low-deg} applies to the above \emph{random $k$-CSP refutation} problem, so long as $\obj \le (1-\delta-\eps)m$ for any $\eps \ge \frac{c\cdot d \log n}{D - 3d}$, where $c$ is a fixed constant.
\end{lemma}
\begin{proof}
    In this case, the instance degree of the \sos relaxation $k = 1$.
    We have from \pref{cor:well-cond} that the degree-$d$ \sos refutation is well-conditioned, with numbers bounded by $n^{c_1 d}$ for some constant $c_1$.
    Define $B = c_1 d$.

    Our subsampling distribution $\Theta$ is the distribution given by including each constraint independently with probability $\rho$, producing an induced CSP instance on $n$ variables with approximately $\rho m$ constraints.
    Since each constraint survives the subsampling with probability $\rho$, for any $\alpha \in \binom{C}{D - 3d}$,
    \[
	\Pr_{S\sim \subsetdist}[\alpha \subseteq S] \le \rho^{D-3d}.
    \]

The subsample inference property clearly holds for the boolean constraints $\{x_i^2 = 1\}_{i\in[n]}$, as a Boolean assignment to the variables is valid regardless of the number of constraints.
    Before subsampling there are at least $(1-\delta)m$ satisfied constraints, and so letting $O_S$ be the number of constraints satisfied in sub-instance $S$, we have by a Chernoff bound
    \[
	\Pr_{S \sim \Theta}[O_S \ge (1-\beta)\cdot \rho(1-\delta)m] \ge 1 - \exp\left(-\frac{\beta^2\rho(1-\delta)m}{2}\right)\mper
    \]
Choosing $\beta = \sqrt{\frac{10B\log n}{\rho(1-\delta)m}} = o(1)$ (with overwhelming probability since we have $\alpha \ge 1\implies \E[m] \ge n$), we have that $\Theta$ gives us $n^{-10B}$-robust inference for the random CSP refutation problem, so long as $\obj \le (1-o(1))\rho(1-\delta)m$.
    Choosing $\rho = (1-\eps)$ so that
    \[
	\rho^{D-3d} \le n^{-8B} \implies \eps \ge \frac{c_2 d\log n}{D-3d},
    \]
for some constant $c_2$.
	The conclusion follows
(after making appropriate adjustments to the constant).
\end{proof}

\begin{problem}[Community detection with average degree $d$ (stochastic block model)]
    Given a graph $G = (V,E)$ on $n$ vertices, determine whether it comes from:
    \begin{compactitem}
    \item {\bf Uniform Distribution}: $G(n,b/n)$, the distribution over graphs in which each edge is included independently with probability $b/n$.
    \item {\bf Planted Distribution}: the stochastic block model---there is a partition of the vertices into two equally-sized sets, $Y$ and $Z$, and the edge $(u,v)$ is present with probability $a/n$ if $u,v \in Y$ or $u,v \in Z$, and with probability $(b-a)/n$ otherwise.
    \end{compactitem}
    Letting $x_1,\ldots,x_n$ be variables corresponding to the membership of each vertex's membership, and let $A$ be the adjacency of the graph.
    The canonical polynomial optimization problem is
    \begin{align*}
	\obj &\le x^\top A x\\
	x_i^2 &= 1 \qquad \forall i \in [n]\\
	\sum_i x_i &= 0.
    \end{align*}
\end{problem}
\begin{lemma}\label{lem:sbm-ex}
    \pref{thm:low-deg} applies to the \emph{community detection} problem so long as $\obj \le (1-\epsilon)\frac{(2a - b)}{4}n$, for $\epsilon > \frac{c \cdot d \log n}{D-3d}$ where $c$ is a fixed constant.
\end{lemma}
\begin{proof}
    The degree of the \sos relaxation in the instance is $k = 1$.
    Since we have only hypercube and balancedness constraints, we have from \pref{cor:well-cond} that the \sos ideal matrix is well-conditioned, with no number in the \sos refutation larger than $n^{c_1 d}$ for some constant $c_1$.
    Let $B = c_1 d$.

    Consider the solution $x$ which assigns $x_i = 1$ to $i \in Y$ and $x_i = -1$ to $i\in Z$.
    Our subsampling operation is to remove every edge independently with probability $1-\rho$.
    The resulting distribution $\Theta$ and the corresponding restriction of $x$ clearly satisfies the Booleanity and balancedness constraints with probability $1$.
Since each edge is included independently with probability $\rho$, for any $\alpha \in \binom{E}{D - 3d}$,
    \[
	\Pr_{S\sim \subsetdist}[\alpha \subseteq S] \le \rho^{D-3d}.
	\]

    In the sub-instance, the expected value (over the choice of planted instance and over the choice of sub-instance) of the restricted solution $x$ is
    \[
	\frac{\rho a}{n}\cdot \left(\binom{|Y|}{2} + \binom{|Z|}{2}\right)
	-\rho\frac{b-a}{n}\cdot |Y|\cdot |Z|
	= \frac{(2a - b)\rho n}{4} - \rho a,
	\]
	and by a Chernoff bound, the value in the sub instance is within a $(1-\beta)$-factor with probability $1 - n^{-10B}$ for $\beta = \sqrt{\frac{10B\log n}{n}}$.
	On resampling the edges outside the sub-instance from the uniform distribution, this value can only decrease by at most $(1-\rho) (1+\beta) nb/2$ w.h.p over the choice of the outside edges.

	If we set $\rho = (1-\eps (2a-b)/10b)$, then $\rho^{D - 3d} \le n^{-8B}$ for $\epsilon \ge \frac{c_2 (2a-b) \log n}{D-3d}$.
	for some constant $c_2$, while the objective value is at least $(1-\epsilon) \frac{(2a-b)n}{4}$.
	The conclusion follows
(after making appropriate adjustments to the constant).
\end{proof}

\begin{problem}[Densest-$k$-subgraph]
    Given a graph $G = (V,E)$ on $n$ vertices, determine whether it comes from:
    \begin{compactitem}
    \item {\bf Uniform Distribution}: $G(n,p)$.
    \item {\bf Planted Distribution}: A graph from $G(n,p)$ with an instance of $G(k,q)$ planted on a random subset of $k$ vertices, $p < q$.
    \end{compactitem}
\end{problem}
    Letting $A$ be the adjacency matrix, the usual polynomial program for \emph{densest-$k$-subgraph} in variables $x_1,\ldots, x_n$ is:
    \begin{align*}
	\obj &\le x^\top A x\\
	x_i^2 &= x_i\quad \forall i \in [n]\\
	\sum_i x_i &= k
    \end{align*}
\begin{lemma}\label{lem:dks-ex}
    When $k^2(p+q) \gg d\log n$,
    \pref{thm:low-deg} applies to the \emph{densest-k-subgraph} problem with  $\obj \le (1-\eps)(p+q)\binom{k}{2}$ for any $\eps > \frac{c\cdot d \log n}{D - 3d}$ for a fixed constant $c$.
\end{lemma}
\begin{proof}
    The degree of the \sos relaxation in the instance is $k = 1$.
    We have from \pref{cor:well-cond} that the \sos proof has no values larger than $n^{c_1 d}$ for a constant $c_1$; fix $B = c_1d$.

    Our subsampling operation is to include each edge independently with probability $\rho$, and take the subgraph induced by the included edges.
    Clearly, the Booleanity and sparsity constraints are preserved by this subsampling distribution $\Theta$.
Since each edge is included independently with probability $\rho$, for any $\alpha \in \binom{E}{D - 3d}$,
    \[
	\Pr_{S\sim \subsetdist}[\alpha \subseteq S] \le \rho^{D-3d}.
	\]

    Now, the expected objective value (over the instance and the sub-sampling) is at least $\rho(p+q)\binom{k}{2}$, and applying a Chernoff bound, we hace that the probability the sub-sampled instance has value less than $(1-\beta)\rho(p+q)\binom{k}{2}$ is at most $n^{-10B}$ if we choose $\beta = \sqrt{\frac{10B \log n}{\rho(p+q)\binom{k}{2}}}$ (which is valid since we assumed that $d\log n \ll (p+q)k^2$).
    Further, a dense subgraph on a subset of the edges is still dense when more edges are added back, so we have the $n^{-10B}$-robust inference property.

    Thus, choosing $\rho = (1-\eps)$ and setting
    \[
	\rho^{D-3d} \le n^{-8B} \implies \eps \ge \frac{c_2 d \log n}{D - 3d},
    \]
for some constant $c_2$,
which concludes the proof (after making appropriate adjustments to the constant).
\end{proof}

\begin{problem}[Tensor PCA]
    Given an order-$k$ tensor in $(\R^{n})^{\tensor k}$, determine whether it comes from:
    \begin{compactitem}
    \item {\bf Uniform Distribution}: each entry of the tensor sampled independently from $\cN(0,1)$.
    \item {\bf Planted Distribution}: a spiked tensor, $\bT = \lambda \cdot v^{\tensor k} + G$  where $v$ is sampled uniformly from $\{\pm \frac{1}{\sqrt{n}}\}^n$, and where $G$ is a random tensor with  each entry sampled independently from $\cN(0,1)$.
    \end{compactitem}
    Given the tensor $\bT$, the canonical program for the tensor PCA problem in variables $x_1,\ldots, x_n$ is:
    \begin{align*}
	\obj &\le \iprod{x^{\tensor k}, \bT}\\
	\|x\|_2^2 &= 1
    \end{align*}
\end{problem}
\begin{lemma}\label{lem:tpca-ex}
    For $\lambda n^{-\eps} \gg \log n$,
    \pref{thm:low-deg} applies to the \emph{tensor PCA} problem with $\obj \le \lambda n^{-\eps}$ for any $\epsilon \ge \frac{c \cdot d}{D-3d}$ for a fixed constant $c$.
\end{lemma}
\begin{proof}
    The degree of the \sos relaxation in the instance is $k = 1$.
    Since the entries of the noise component of the tensor are standard normal variables, with exponentially good probability over the input tensor $\bT$ we will have no entry of magnitude greater than $n^{d}$.
    This, together with \pref{cor:well-cond}, gives us that except with exponentially small probability the \sos proof will have no values exceeding $n^{c_1 d}$ for a fixed constant $c_1$.

    Our subsampling operation is to set to zero every entry of $\bT$ independently with probability $1-\rho$, obtaining a sub-instance $\bT'$ on the nonzero entries.
    Also, for any $\alpha \in \binom{[n]^k}{D - 3d}$,
    \[
	\Pr_{S\sim \Theta}[\alpha \in S] \le \rho^{D-3d}.
    \]

    This subsampling operation clearly preserves the planted solution unit sphere constraint.
    Additionally, let $\cR$ be the operator that restricts a tensor to the nonzero entries.
    We have that $\iprod{\cR(\lambda \cdot v^{\tensor k}),v^{\tensor k}}$ has expectation $\lambda \cdot \rho$, since every entry of $v^{\tensor k}$ has magnitude $n^{-k/2}$.
    Applying a Chernoff bound, we have that this quantity will be at least $(1-\beta)\lambda \rho$ with probability at least $n^{-10B}$ if we choose $\beta = \sqrt{\frac{10B\log n}{\lambda \rho}}$.

    It remains to address the noise introduced by $G_{\bT'}$ and resampling all the entries outside of the subinstance $\bT'$.  Each of these entries is a standard normal entry.
    The quantity $\iprod{(\Id -\cR)(N),v^{\tensor k}}$ is a sum over at most $n^k$ i.i.d. Gaussian entries each with standard deviation $n^{-k/2}$ (since that is the magnitude of $(v^{\tensor k})_{\alpha}$.
    The entire quantity is thus a Gaussian random variable with mean $0$ and variance $1$, and therefore with probability at least $n^{-10B}$ this quantity will not exceed $\sqrt{10 B \log n}$.
    So long as $\sqrt{10 B\log n} \ll \lambda \rho$, the signal term will dominate, and the solution will have value at least $\lambda\rho/2$.

    Now, we set $\rho = n^{-\eps}$ so that
    \[
	\rho^{D - 3d} \le n^{-8B} \implies \epsilon \ge \frac{2c_1 d}{D-3d},
    \]
which concludes the proof (after making appropriate adjustments to the constant $c_1$).
\end{proof}

\begin{problem}[Sparse PCA]
    Given an $n \times m$ matrix $M$ in $\R^n$, determine whether it comes from:
    \begin{compactitem}
    \item {\bf Uniform Distribution}: each entry of the matrix sampled independently from $\cN(0,1)$.
    \item {\bf Planted Distribution}: a random vector with $k$ non-zero entries $v \in \{0,\pm 1/\sqrt{k}\}^n$ is chosen, and then the $i$th column of the matrix is sampled independently by taking $s_i v + \gamma_i$ for a uniformly random sign $s_i \in \{\pm 1\}$ and a standard gaussian vector $\gamma_i \sim \cN(0,\Id)$.
    \end{compactitem}
    The canonical program for the sparse PCA problem in variables $x_1,\ldots, x_n$ is:
    \begin{align*}
	\obj &\le \|M^\top x\|_2^2\\
	x_i^2 &= x_i \quad \forall i\in [n]\\
	\|x\|_2^2 &= k
    \end{align*}
\end{problem}
\begin{lemma}\label{lem:spca-ex}
    For $k n^{-\eps/2} \gg \log n$,
    \pref{thm:low-deg} applies to the \emph{sparse PCA} problem with $\obj \le k^{2-\eps}m$ for any $\epsilon > \frac{c \cdot d}{D-6d}$ for a fixed constant $c$.
\end{lemma}
\begin{proof}
    The degree of the \sos relaxation in the instance is $2$.
    Since the entries of the noise are standard normal variables, with exponentially good probability over the input matrix $M$ we will have no entry of magnitude greater than $n^{d}$.
    This, together with \pref{cor:well-cond}, gives us that except with exponentially small probability the \sos proof will have no values exceeding $n^{c_1 d}$ for a fixed constant $c_1$.

    Our subsampling operation is to set to zero every entry of $M$ independently with probability $1-\rho$, obtaining a sub-instance $M$ on the nonzero entries.
    Also, for any $\alpha \in \binom{M}{D - 6d}$,
    \[
	\Pr_{S\sim \Theta}[\alpha \in S] \le \rho^{D-6d}.
    \]

    This subsampling operation clearly preserves the constraints on the solution variables.

    We take our subinstance solution $y = \sqrt{k} v$, which is feasible.
    Let $\cR$ be the subsampling operator that zeros out a set of columns.  On subsampling, and then resampling the zeroed out columns from the uniform distribution, we can write the resulting $\tilde{M}$ as
    \[  \tilde{M}^{\top} = \cR(sv^{T}) + G^{\top}\]
    where $G^T$ is a random Gaussian matrix.  Therefore, the objective value obtained by the solution $y = \sqrt{k} v$ is

    \[
	    \tilde{M}^\top y = \sqrt{k} \cdot \cR(sv^\top)v + \sqrt{k}\cdot G^\top v
    \]
The first term is a vector $u_{signal}$ with $m$ entries, each of which is a sum of $k$ Bernoulli random variables, all of the same sign, with probability $\rho$ of being nonzero.
    The second term is a vector $u_{noise}$ with $m$ entries, each of them an independent Gaussian variable with variance bounded by $ k$.
    We have that
    \[
	\E_{\Theta}[\|u_{signal}\|_2^2] = (\rho k)^2 m,
    \]
and by Chernoff bounds we have that this concentrates within a $(1-\beta)$ factor with probability $1-n^{-10B}$ if we take $\beta = \sqrt{\frac{10B\log n}{(\rho k)^2 m}}$.

    The expectation of $\iprod{u_{signal},u_{noise}}$ is zero, and applying similar concentration arguments \Tnote{} we have that with probability $1-n^{10B}$, $|\iprod{u_{signal},u_{noise}}| \le (1+\beta)\rho k$.
    Taking the union bound over these events and applying Cauchy-Schwarz, we have that
    \[
	\|\cR(M)y\|_2^2 \ge (\rho k)^2 m - 2 k m = \rho^2 k^2 m - 2 k m.
    \]
so long as $\rho k \gg 1$, the first term dominates.

    Now, we set $\rho = n^{-\eps}$ for $\eps < 1$ so that
    \[
	\rho^{D - 6d} \le n^{-8B} \implies \epsilon \ge \frac{c_2 d}{D-6d},
    \]
for some constant $c_2$, which concludes the proof.
\end{proof}

\begin{remark} For tensor PCA and sparse PCA, the underlying distributions were Gaussian.  Applying \pref{thm:main} in these contexts yields the existence of distinguishers that are {\it low-degree} in a non-standard sense.  Specifically, the degree of a monomial will be the number of distinct variables in it, irrespective of the powers to which they are raised.
\end{remark}
\section{Exponential lower bounds for PCA problems}
\label{sec:pca}
In this section we give an overview of the proofs of our \sos lower bounds for the tensor and sparse PCA problems.
We begin by showing how Conjecture \ref{conj:main-conjecture} predicts such a lower bound in the tensor PCA setting.
Following this we state the key lemmas to prove the exponential lower bounds; since these lemmas can be proved largely by techniques present in the work of Barak et al. on planted clique \cite{DBLP:conf/focs/BarakHKKMP16}, we leave the details to a forthcoming full version of the present paper.
\Snote{}

\subsection{Predicting sos lower bounds from low-degree distinguishers for Tensor PCA}
In this section we demonstrate how to predict using Conjecture~\ref{conj:main-conjecture} that when $\lambda \ll n^{3/4 - \e}$ for $\e > 0$, \sos algorithms cannot solve Tensor PCA.
This prediction is borne out in Theorem~\ref{thm:tpca-intro}.

\begin{theorem}
Let $\mu$ be the distribution on $\R^{n \tensor n \tensor n}$ which places a standard Gaussian in each entry.
  Let $\nu$ be the density of the Tensor PCA planted distribution with respect to $\mu$, where we take the planted vector $v$ to have each entry uniformly chosen from $\{ \pm \tfrac 1 {\sqrt n} \}$.\footnote{This does not substantially modify the problem but it will make calculations in this proof sketch more convenient.}
  If $\lambda \leq n^{3/4 - \e}$, there is no degree $n^{o(1)}$ polynomial $p$ with
  \[
  \E_\mu p(A) = 0, \quad \E_{\text{planted}} p(A) \geq n^{\Omega(1)} \cdot \Paren{\Var_{\mu} p(A)}^{1/2}\mper
  \]
\end{theorem}
We sketch the proof of this theorem.
The theorem follows from two claims.
\begin{claim}\label{clm:tpca-sketch-1}
\begin{align}\label{eq:tpca-predict}
  \max_{\substack{\deg p \leq n^{o(1)} \\, \E_\mu p(T) = 0} } \frac{\E_{ \nu} p(T)}{\Paren{\E_\mu p(T)^2}^{1/2}} = (\E_\mu (\nu^{\leq d}(T) - 1)^2)^{1/2}
\end{align}
  where $\nu^{\leq d}$ is the orthogonal projection (with respect to $\mu$) of the density $\nu$ to the degree-$d$ polynomials.
  Note that the last quantity is just the $2$ norm, or the variance, of the truncation to low-degree polynomials of the density $\nu$ of the planted distribution.
\end{claim}

\begin{claim}\label{clm:tpca-sketch-2}
  $(\E_\mu (v^{\leq d}(T) - 1)^2)^{1/2} \ll 1$ when $\lambda \leq n^{3/4 - \e}$ for $\e \geq \Omega(1)$ and $d = n^{o(1)}$.
\end{claim}

The theorem follows immediately.
We sketch proofs of the claims in order.

\begin{proof}[Sketch of proof for Claim~\ref{clm:tpca-sketch-1}]
By definition of $\nu$, the maximization is equivalent to maximizing $\E_\mu \nu(T) \cdot p(T)$ among all $p$ of degree $d =n^{o(1)}$ and with $\E_\mu p(T)^2 = 1$ and $\E_\mu p(T) = 0$.
Standard Fourier analysis shows that this maximum is achieved by the orthogonal projection of $\nu - 1$ into the span of degree $d$ polynomials.

To make this more precise, recall that the Hermite polynomials provide an orthonormal basis for real-valued polynomials under the multivariate Gaussian distribution.
(For an introduction to Hermite polynomials, see the book \cite{DBLP:books/daglib/0033652}.)
The tensor $T \sim \mu$ is an $n^3$-dimensional multivariate Gaussian.
For a (multi)-set $W \subseteq [n]^3$, let $H_W$ be the $W$-th Hermite polynomial, so that $\E_\mu H_W(T) H_{W'}(T) = \Ind_{W = W'}$.

Then the best $p$ (ignoring normalization momentarily) will be the function
\[
  p(A) = \nu^{\leq d}(A)-1 = \sum_{1 \leq |W| \leq d} (\E_{T \sim \mu} \nu(T)H_W(T)) \cdot H_W(A)
\]
Here $\E_\mu \nu(T) H_W(T) = \widehat{\nu}(W)$ is the $W$-th Fourier coefficient of $\nu$.
What value for \eqref{eq:tpca-predict} is achieved by this $p$?
Again by standard Fourier analysis, in the numerator we have,
\[
  \E_\nu p(T) = \E_\nu (\nu^{\leq D}(T)-1) = \E_\mu \nu(T) \cdot (\nu^{\leq D}(T)-1) = \E_\mu (\nu^{\leq d}(T)-1)^2
\]
Comparing this to the denominator, the maximum value of \eqref{eq:tpca-predict} is $(\E_\mu (v^{\leq d}(T)-1)^2)^{1/2}$.
This is nothing more than the $2$-norm of the projection of $\nu - 1$ to degree-$d$ polynomials!
\end{proof}

The following fact, used to prove Claim~\ref{clm:tpca-sketch-2}, is an elementary computation with Hermite polynomials.
\begin{fact}
  Let $W \subseteq [n]^3$.
  Then $\widehat{\nu}(W) = \lambda^{|W|} n^{-3|W|/2} $ if $W$, thought of as a $3$-uniform hypergraph, has all even degrees, and is $0$ otherwise.
\end{fact}
To see that this calculation is straightforward, note that
$\widehat{\nu(W)} = \E_\mu \nu(T) H_W(T) = \E_{\nu} H_W(T)$,
so it is enough to understand the expectations of the Hermite polynomials under the planted distribution.

\begin{proof}[Sketch of proof for Claim~\ref{clm:tpca-sketch-2}]
  Working in the Hermite basis (as described above), we get $\E_\mu (v^{\leq d}(T)-1)^2 = \sum_{1 \leq |W| \leq d} \widehat{\nu}(W)^2$.
For the sake of exposition, we will restrict attention in the sum to $W$ in which no element appears with multiplicity larger than $1$ (other terms can be treated similarly).

What is the contribution to $\sum_{1 \leq |W| \leq d} \widehat{\nu}(W)^2$ of terms $W$ with $|W| = t$?
By the fact above, to contribute a nonzero term to the sum, $W$,considered as a $3$-uniform hypergraph must have even degrees.
So, if it has $t$ hyperedges, it contains at most $3t/2$ nodes.
  There are $n^{3t/2}$ choices for these nodes, and having chosen them, at most $t^{O(t)}$ $3$-uniform hypergraphs on those nodes.
  Hence,
  \[
    \sum_{1 \leq |W| \leq d} \widehat{\nu}(W)^2 \leq \sum_{t = 1}^d n^{3t/2} t^{O(t)} \lambda^{2t} n^{-3t}\mper
  \]
So long as $\lambda^2 \leq n^{3/2 - \e}$ for some $\e = \Omega(1)$ and $t \leq d \leq n^{O(\epsilon)}$, this is $o(1)$.
\end{proof}

\subsection{Main theorem and proof overview for Tensor PCA}\label{sec:tpca}
In this section we give an overview of the proof of Theorem~\ref{thm:tpca-intro}.
The techniques involved in proving the main lemmas are technical refinements of techniques used in the work of Barak et al. on \sos lower bounds for planted clique \cite{DBLP:conf/focs/BarakHKKMP16}; we therefore leave full proofs to a forthcoming full version of this paper.

To state and prove our main theorem on tensor PCA it is useful to define a Boolean version of the problem.
For technical convenience we actually prove an \sos lower bound for this problem; then standard techniques (see Section~\ref{sec:reduce-to-hypercube}) allow us to prove the main theorem for Gaussian tensors.
\begin{problem}[$k$-Tensor PCA, signal-strength $\lambda$, boolean version]\label{prob:tpca}
  Distinguish the following two distributions on $\Omega_k \defeq \{ \pm 1\}^{n \choose k}$.
  \begin{itemize}
    \item \emph{the uniform distribution: } $A \sim \Omega$ chosen uniformly at random.
    \item \emph{the planted distribution: } Choose $v \sim \{\pm 1\}^n$ and let $B = v^{\tensor k}$.
      Sample $A$ by rerandomizing every coordinate of $B$ with probability $1 - \lambda n^{-k/2}$.
  \end{itemize}
\end{problem}

We show that the natural SoS relaxation of this problem suffers from a large integrality gap, when $\lambda$ is slightly less than  $n^{k/4}$, even when the degree of the SoS relaxation is $n^{\Omega(1)}$.
(When $\lambda \gg n^{k/4 - \e}$, algorithms with running time $2^{n^{O(\e)}}$ are known for $k = O(1)$ \cite{DBLP:conf/nips/RichardM14, DBLP:conf/colt/HopkinsSS15, DBLP:conf/stoc/HopkinsSSS16, DBLP:journals/corr/BhattiproluGL16, DBLP:journals/corr/RaghavendraRS16}.)

\begin{theorem}\label{thm:tpca-main}
  Let $k = O(1)$.
  For $A \in \Omega_k$, let
  \[
    SoS_d(A) \defeq \max_{\pE} \, \pE \iprod{x^{\tensor k}, A} \text{ s.t. $\pE$ is a degree-$d$ pseudoexpectation satisfying $\{ \|x\|^2 = 1 \}$}\mper
  \]
  There is a constant $c$ so that for every small enough $\epsilon > 0$, if $d \leq n^{c \cdot \epsilon}$, then for large enough $n$,
  \[
    \Pr_{A \sim \Omega} \{ SoS_d(A) \geq n^{k/4 - \epsilon} \} \geq 1 - o(1)
  \]
  and
  \[
    \E_{A \sim \Omega} SoS_d(A) \geq n^{k/4 - \epsilon}\mper
  \]
  Moreover, the latter also holds for $A$ with iid entries from $\cN(0,1)$.\footnote{For technical reasons we do not prove a tail bound type statement for Gaussian $A$, but we conjecture that this is also true.}
\end{theorem}

To prove the theorem we will exhibit for a typical sample $A$ from the uniform distribution a degree $n^{\Omega(\e)}$ pseudodistribution $\pE$ which satisfies $\{\|x\|^2 = 1\}$ and has $\pE \iprod{x^{\tensor k}, A} \geq n^{k/4 - \epsilon}$.
The following lemma ensures that the pseudo-distribution we exhibit will be PSD.

\begin{lemma}\label{lem:tpca-conditions-main}
  Let $d \in \N$ and let $N_d = \sum_{s \leq d} n (n-1)\cdots(n-(s-1))$ be the number of $\leq d$-tuples with unique entries from $[n]$.
  There is a constant $\e^*$ independent of $n$ such that for any $\e < \e^*$ also independent of $n$, the following is true.
  Let $\lambda = n^{k/4 - \e}$.
  Let $\mu(A)$ be the density of the following distribution (with respect to the uniform distribution on $\Omega = \{\pm 1\}^{\binom nk}$).

  \textbf{The Planted Distribution: } Choose $v \sim \{\pm 1 \}^n$ uniformly.
  Let $B = v^{\tensor k}$.
  Sample $A$ by
  \begin{itemize}
    \item replacing every coordinate of $B$ with a random draw from $\{\pm 1\}$ independently with probability $1 - \lambda n^{-k/2}$,
    \item then choosing a subset $S \subseteq [n]$ by including every coordinate with probability $n^{-\epsilon}$,
    \item then replacing every entry of $B$ with some index outside $S$ independently with a uniform draw from $\{\pm 1\}$.
  \end{itemize}
  Let $\Lambda : \Omega \rightarrow \R^{N_d \times N_d}$ be the following function
  \[
    \Lambda(A) = \mu(A) \cdot \E_{v |A} v^{\tensor \leq 2d}
  \]
  Here we abuse notation and denote by $x^{ \leq \tensor 2d}$ the matrix indexed by tuples of length $ \leq d$ \emph{with unique entries} from $[n]$. 
For $D \in \N$, let $\Lambda^{\leq D}$ be the projection of $\Lambda$ into the degree-$D$ real-valued polynomials on $\{\pm 1\}^{n \choose k}$.
  There is a universal constant $C$ so that if $C d/\epsilon < D < n^{\epsilon/C}$, then for large enough $n$
  \[
    \Pr_{A \sim \Omega} \{ \Lambda^{\leq D}(A) \succeq 0 \} \geq 1 - o(1)\mper
  \]
\end{lemma}

For a tensor $A$, the moment matrix of the pseudodistribution we exhibit will be $\Lambda^{\leq D}(A)$.
We will need it to satisfy the constraint $\{\|x\|^2 = 1 \}$.
This follows from the following general lemma.
(The lemma is much more general than what we state here, and uses only the vector space structures of space of real matrices and matrix-valued functions.)
\begin{lemma}\label{lem:trunc-constraints}
  Let $k \in \N$.
  Let $V$ be a linear subspace of $\R^{N \times M}$.
  Let $\Omega = \{\pm 1\}^{\binom{n}{k}}$.
  Let $\Lambda : \Omega \rightarrow V$.
  Let $\Lambda^{\leq D}$ be the entrywise orthogonal projection of $\Lambda$ to polynomials of degree at most $D$.
  Then for every $A \in \Omega$, the matrix $\Lambda^{\leq D} (A) \in V$.
\end{lemma}
\begin{proof}
  The function $\Lambda$ is an element of the vector space $\R^{N \times M} \tensor \R^{\Omega}$.
  The projection $\Pi_V : \R^{N \times M} \rightarrow V$ and the projection $\Pi_{\leq D}$ from $\R^{\Omega}$ to the degree-$D$ polynomials commute as projections on $\R^{N \times M} \tensor \R^{\Omega}$, since they act on separate tensor coordinates.
  It follows that $\Lambda^{\leq D} \in V \tensor (\R^\Omega)^{\leq D}$ takes values in $V$.
\end{proof}

Last, we will require a couple of scalar functions of $\Lambda^{\leq D}$ to be well concentrated.
\begin{lemma}\label{lem:tpca-scalar-concentration}
  \Snote{}
  Let $\Lambda,d,\e,D$ be as in Lemma~\ref{lem:tpca-conditions-main}.
  The function $\Lambda^{\leq D}$ satisfies
  \begin{itemize}
    \item $\Pr_{A \sim \Omega} \{ \Lambda^{\leq D}_{\emptyset, \emptyset}(A) = 1 \pm o(1) \} \geq 1 - o(1)$ (Here $\Lambda_{\emptyset, \emptyset} = 1$ is the upper-left-most entry of $\Lambda$.)
    \item $\Pr_{A \sim \Omega} \{ \iprod{\Lambda^{\leq D}(A), A} = (1 \pm o(1))\cdot n^{3k/4 - \epsilon} \} \geq 1 - o(1)$ (Here we are abusing notation to write $\iprod{\Lambda^{\leq D}(A), A}$ for the inner product of the part of $\Lambda^{\leq D}$ indexed by monomials of degree $k$ and $A$.)
  \end{itemize}
\end{lemma}

The Boolean case of Theorem~\ref{thm:tpca-main} follows from combining the lemmas.
The Gaussian case can be proved in a black-box fashion from the Boolean case following the argument in Section~\ref{sec:reduce-to-hypercube}.

The proofs of all the lemmas in this section follow analogous lemmas in the work of Barak et al. on planted clique \cite{DBLP:conf/focs/BarakHKKMP16}; we defer them to the full version of the present work.

\subsection{Main theorem and proof overview for sparse PCA}
\label{sec:spca-main}
In this section we prove the following main theorem.
Formally, the theorem shows that with high probability for a random $n \times n$ matrix $A$, even high-degree \sos relaxations are unable to certify that no sparse vector $v$ has large quadratic form $\iprod{v,Av}$.

\begin{theorem}[Restatement of Theorem~\ref{thm:spca-main}]
  If $A \in \R^{n \times n}$, let
  \[
    SoS_{d,k}(A) = \max_{\pE} \pE \iprod{x,Ax} \text{ s.t. $\pE$ is degree $d$ and satisfies }\left \{ x_i^3 = x_i, \|x\|^2 = k \right \} \mper
  \]
  There are absolute constants $c,\e^* > 0$ so that for every $\rho \in (0,1)$ and $\e \in (0,\e^*)$, if $k = n^{\rho}$, then for $d \leq n^{c \cdot \e}$,
  \[
    \Pr_{A \sim \{\pm 1\}^{\binom{n}{2}}} \{ SoS_{d,k}(A) \geq \min (n^{1/2 - \epsilon} k, n^{\rho - \e} k)  \} \geq 1 - o(1)
  \]
  and
  \[
    \E_{A \sim \{\pm 1\}^{\binom{n}{2}}}  SoS_{d,k}(A) \geq \min (n^{1/2 - \epsilon} k, n^{\rho - \e} k)\mper
  \]
  Furthermore, the latter is true also if $A$ is symmetric with iid entries from $\cN(0,1)$.\footnote{For technical reasons we do not prove a tail bound type statement for Gaussian $A$, but we conjecture that this is also true.}
\end{theorem}
We turn to some discussion of the theorem statement.
First of all, though it is technically convenient for $A$ in the theorem statement above to be a $\pm 1$ matrix, the entries may be replaced by standard Gaussians (see Section~\ref{sec:reduce-to-hypercube}).

\begin{remark}[Relation to the spiked-Wigner model of sparse principal component analysis]
To get some intuition for the theorem statement, it is useful to return to a familiar planted problem: the spiked-Wigner model of sparse principal component analysis.
Let $W$ be a symmetric matrix with iid entries from $\cN(0,1)$, and let $v$ be a random $k$-sparse unit vector with entries $\{\pm 1/\sqrt{k}, 0\}$.
Let $B = W + \lambda \dyad{v}$.
The problem is to distinguish between a single sample from $B$ and a sample from $W$.
There are two main algorithms for this problem, both captured by the \sos hierarchy.
The first, applicable when $\lambda \gg \sqrt n$, is vanilla PCA: the top eigenvalue of $B$ will be larger than the top eigenvalue of $W$.
The second, applicable when $\lambda \gg k$, is diagonal thresholding: the diagonal entries of $B$ which corresponds to nonzero coordinates will be noticeably large.
The theorem statement above (transferred to the Gaussian setting, though this has little effect) shows that once $\lambda$ is well outside these parameter regimes, i.e. when $\lambda < n^{1/2 - \e}, k^{1 - \e}$ for arbitrarily small $\e > 0$, even degree $n^{\Omega(\e)}$ \sos programs do not distinguish between $B$ and $W$.
\end{remark}

\begin{remark}[Interpretation as an integrality gap]
A second interpretation of the theorem statement, independent of any planted problem, is as a strong integrality gap for random instances for the problem of maximizing a quadratic form over $k$-sparse vectors.
Consider the actual maximum of $\iprod{x,Ax}$ for random ($\{ \pm 1\}$ or Gaussian) $A$ over $k$-sparse unit vectors $x$.
There are roughly $2^{k \log n}$ points in a $\tfrac 1 2$-net for such vectors, meaning that by standard arguments,
\[
  \max_{\|x\| = 1, x \text{ is $k$-sparse}} \iprod{x,Ax} \leq O(\sqrt k \log n)\mper
\]
\Snote{}
With the parameters of the theorem, this means that the integrality gap of the degree $n^{\Omega(\e)}$ \sos relaxation is at least $\min (n^{\rho/2 - \e}, n^{1/2 - \rho/2 - \e})$ when $k = n^{\rho}$.
\end{remark}

\begin{remark}[Relation to spiked-Wishart model]
Theorem~\ref{thm:spca-main} most closely concerns the spiked-Wigner model of sparse PCA; this refers to independence of the entries of the matrix $A$.
Often, sparse PCA is instead studied in the (perhaps more realistic) \emph{spiked-Wishart model}, where the input is $m$ samples $x_1,\ldots,x_m$ from an $n$-dimensional Gaussian vector $\cN(0, \Id + \lambda \cdot vv^\top)$, where $v$ is a unit-norm $k$-sparse vector.
Here the question is: as a function of the sparsity $k$, the ambient dimension $n$, and the signal strength $\lambda$, how many samples $m$ are needed to recover the vector $v$?
\Snote{}
The natural approach to recovering $v$ in this setting is to solve a convex relaxation of the problem of maximizing he quadratic form of the empirical covariance $M = \sum_{i \leq m} \dyad{x_i}$ over $k$-sparse unit vectors (the maximization problem itself is NP-hard even to approximate \cite{DBLP:conf/colt/ChanPR16}).

Theoretically, one may apply our proof technique for Theorem~\ref{thm:spca-main} directly to the spiked-Wishart model, but this carries the expense of substantial technical complication.
We may however make intelligent guesses about the behavior of \sos relaxations for the spiked-Wishart model on the basis of Theorem~\ref{thm:spca-main} alone.
As in the spiked Wigner model, there are essentially two known algorithms to recover a planted sparse vector $v$ in the spiked Wishart model: vanilla PCA and diagonal thresholding \cite{DBLP:conf/nips/DeshpandeM14}.
We conjecture that, as in the spiked Wigner model, the \sos hierarchy requires $n^{\Omega(1)}$ degree to improve the number of samples required by these algorithms by any polynomial factor.
Concretely, considering the case $\lambda = 1$ for simplicity, we conjecture that there are constants $c, \e^*$ such that for every $\e \in (0,\e^*)$ if $m \leq \min(k^{2 - \e}, n^{1 - \e})$ and $x_1,\ldots,x_m \sim \cN(0, \Id)$ are iid, then with high probability for every $\rho \in (0,1)$ if $k = n^\rho$,
\[
  \sos_{d,k}\Paren{\sum_{i \leq m} \dyad{x_i}} \geq \min(n^{1-\e} k, k^{2 - \e})
\]
for all $d \leq n^{c \cdot \e}$.
\end{remark}

\paragraph{Lemmas for Theorem~\ref{thm:spca-main}}
Our proof of Theorem~\ref{thm:spca-main} is very similar to the analogous proof for Tensor PCA, Theorem~\ref{thm:tpca-main}.
We state the analogues of \pref{lem:tpca-conditions-main} and \pref{lem:tpca-scalar-concentration}.
\pref{lem:trunc-constraints} can be used unchanged in the sparse PCA setting.

The main lemma, analogous to \pref{lem:tpca-conditions-main} is as follows.
\begin{lemma}\label{lem:spca-psd}
   Let $d \in \N$ and let $N_d = \sum_{s \leq d} n (n-1)\cdots(n-(s-1))$ be the number of $\leq d$-tuples with unique entries from $[n]$.
   Let $\mu(A)$ be the density of the following distribution on $n \times n$ matrices $A$ with respect to the uniform distribution on $\{\pm 1\}^{n \choose 2}$.

  \textbf{Planted distribution: }
  Let $k = k(n) \in \N$ and $\lambda = \lambda(n) \in \R$, and $\gamma > 0$, and assume $\lambda \leq k$.
  Sample a uniformly random $k$-sparse vector $v \in \R^n$ with entries $\pm 1, 0$.
  Form the matrix $B = vv^\top$.
  For each nonzero entry of $B$ independently, replace it with a uniform draw from $\{ \pm 1\}$ with probability $1 - \lambda/k$ (maintaining the symmetry $B = B^\top$).
  For each zero entry of $B$, replace it with a uniform draw from $\{ \pm 1 \}$ (maintaining the same symmetry).
  Finally, choose every $i \in [n]$ with probability $n^{-\gamma}$ independently; for those indices that were not chosen, replace every entry in the corresponding row and column of $B$ with random $\pm 1$ entries.\footnote{This additional $n^{-\gamma}$ noising step is a technical convenience which has the effect of somewhat decreasing the number of nonzero entries of $v$ and decreasing the signal-strength $\lambda$.}
  Output the resulting matrix $A$.
  (We remark that this matrix is a Boolean version of the more standard spiked-Wigner model $B + \lambda vv^\top$ where $B$ has iid standard normal entries and $v$ is a random $k$-sparse unit vector with entries from $\{ \pm 1/\sqrt{k}, 0 \}$.)

  Let $\Lambda : \{\pm 1\}^{\binom{n}{2}} \rightarrow \R^{N_d \times N_d}$ be the following function
  \[
    \Lambda(A) = \mu(A) \cdot \E_{v |A} v^{\tensor \leq 2d}
  \]
  where the expectation is with respect to the planted distribution above.
  For $D = D(n) \in \N$, let $\Lambda^{\leq D}$ be the entrywise projection of $\Lambda$ into the Boolean functions of degree at most $D$.

  There are constants $C ,\e^* > 0$ such that for every $\gamma > 0$ and $\rho \in (0,1)$ and every $\e \in (0,\e^*)$ (all independent of $n$), if $k = n^{\rho}$ and $\lambda \leq \min\{n^{\rho - \e}, n^{1/2 - \e} \}$, and if $C d/\epsilon < D < n^{\epsilon/C}$, then for large enough $n$
  \[
    \Pr_{A \sim \{\pm 1\}^{\binom{n}{2}}} \{ \Lambda^{\leq D}(A) \succeq 0 \} \geq 1 - o(1)\mper
  \]
\end{lemma}
\begin{remark}
We make a few remarks about the necessity of some of the assumptions above.
A useful intuition is that the function $\Lambda^{\leq D}(A)$ is (with high probability) positive-valued when the parameters $\rho,\epsilon,\gamma$ of the planted distribution are such that there is no degree-$D$ polynomial $f \, : \, \{ \pm 1\}^{\binom{n}{2}} \rightarrow \R$ whose values distinguish a typical sample from the planted distribution from a null model: a random symmetric matrix with iid entries.

At this point it is useful to consider a more familiar planted model, which the lemma above mimics.
Let $W$ be a $n \times n$ symmetric matrix with iid entries from $\cN(0,1)$.
Let $v \in \R^n$ be a $k$-sparse unit vector, with entries in $\{\pm 1/\sqrt{k}, 0\}$.
Let $A = W + \lambda \dyad{v}$.
Notice that if $\lambda \gg k$, then diagonal thresholding on the matrix $W$ identifies the nonzero coordinates of $v$.
(This is the analogue of the covariance-thresholding algorithm in the spiked-Wishart version of sparse PCA.)
On the other hand, if $\lambda \gg \sqrt n$ then (since typically $\|W\| \approx \sqrt n$), ordinary PCA identifies $v$.
The lemma captures computational hardness for the problem of distinguishing a single sample from $A$ from a sample from the null model $W$ both diagonal thresholding and ordinary PCA fail.
\end{remark}

Next we state the analogue of \pref{lem:tpca-scalar-concentration}.
\begin{lemma}\label{lem:spca-scalar-concentration}
  Let $\Lambda,d,k,\lambda,\gamma,D$ be as in \pref{lem:spca-psd}.
  The function $\Lambda^{\leq D}$ satisfies
  \begin{itemize}
    \item $\Pr_{A \sim \{\pm 1\}^{\binom nk} } \{ \Lambda^{\leq D}_{\emptyset, \emptyset}(A) = 1 \pm o(1) \} \geq 1 - o(1)$.
    \item $\Pr_{A \sim \{\pm 1\}^{\binom nk} } \{ \iprod{\Lambda^{\leq D}(A), A} = (1 \pm o(1))\cdot \lambda n^{\Theta(-\gamma)} \} \geq 1 - o(1)$.
  \end{itemize}
\end{lemma}

\addreferencesection
\bibliographystyle{amsalpha}
\bibliography{bib/mathreview,bib/dblp,bib/scholar,bib/custom}

\appendix
\section{Bounding the sum-of-squares proof ideal term}\label{app:idcond}
We give conditions under which sum-of-squares proofs are well-conditioned, using techniques similar to those that appear in \cite{DBLP:journals/corr/RaghavendraW17} for bounding the bit complexity of \sos proofs.
We begin with some definitions.
\begin{definition}
Let $\cP$ be a polynomial optimization problem and let $\cD$ be the uniform distribution over the set of feasible solutions $S$ for $\cP$.
Define the degree-$2d$ moment matrix of $\cD$ to be $X_{\cD} = \E_{s\sim\cD}[\hat{s}^{\tensor 2d}]$, where $\hat{s} = [1 \ s]^{\top}$.
\begin{itemize}
\item We say that \emph{$\cP$ is $k$-complete on up to degree $2d$} if every zero eigenvector of $X_{\cD}$ has a degree-$k$ derivation from the ideal constraints of $\cP$.
\end{itemize}
\end{definition}

\begin{theorem}\label{thm:idealthing}
    Let $\cP$ be a polynomial optimization problem over variables $x \in \R^n$ of degree at most $2d$, with objective function $f(x)$ and ideal constraints $\{g_j(x) = 0\}_{j \in [m]}$.
    Suppose also that $\cP$ is $2d$-complete up to degree $2d$.
    Let $G$ be the matrix of ideal constraints in the degree-$2d$ \sos proof for $\cP$.
    Then if
    \begin{compactitem}
    \item the SDP optimum value is bounded by $n^{O(d)}$
    \item the coefficients of the objective function are bounded by $n^{O(d)}$,
    \item there is a set of feasible solutions $\cS \subseteq \R^n$ with the property that for each $\alpha \subseteq [n]^d$, $|\alpha| \le d$ for which $\chi_{\alpha}$ is not identically zero over the solution space, there exists some $s \in \cS$ such that the square monomial $\chi_\alpha(s)^2 \ge n^{-O(d)}$,
    \end{compactitem}
    it follows that the \sos certificate for the problem is well-conditioned, with no value larger than $n^{O(d)}$.
\end{theorem}
To prove this, we essentially reproduce the proof of the main theorem of \cite{DBLP:journals/corr/RaghavendraW17}, up to the very end of the proof at which point we slightly deviate to draw a different conclusion.
\begin{proof}
Following our previous convention, the degree-$2d$ sum-of-squares proof for $\cP$ is of the form
\begin{align*}
\sdpopt - f(x)=  a(x) + g(x),
\end{align*}
where the $g(x)$ is a polynomial in the span of the ideal constraints, and $A$ is a sum of squares of polynomials.
Alternatively, we have the matrix characterization,
\[
\sdpopt - \iprod{F, \hat x^{\tensor 2d}} = \iprod{A,\hat x^{\tensor 2d}} + \iprod{G, \hat x^{\tensor 2d}},
\]
where $\hat x = [1 \ x]^{\top}$, $F,A$, and $G$ are matrix polynomials corresponding to $f, a$, and $g$ respectively, and with $A \sge 0$.

Now let $s \in \cS$ be a feasible solution.
Then we have that
\begin{align*}
    \sdpopt - \iprod{F, s^{\tensor 2d}}&= \iprod{A,s^{\tensor 2d}} + \iprod{G, s^{\tensor 2d}}
= \iprod{A,s^{\tensor 2d}},
\end{align*}
where the second equality follows because each $s \in S$ is feasible.
    By assumption the left-hand-side is bounded by $n^{O(d)}$.

    We will now argue that the diagonal entries of $A$ cannot be too large.
    Our first step is to argue that $A$ cannot have nonzero diagonal entries unless there is a solution element in the solution
    Let $X_{\cD} = \E[x^{\tensor 2d}]$ be the $2d$-moment matrix of the uniform distribution of feasible solutions to $\cP$.
    Define $\Pi$ to be the orthogonal projection into the zero eigenspace of $X_{\cD}$.
By linearity and orthonormality, we have that
\begin{align*}
\Iprod{X_{\cD}, A}
&= \Iprod{X_{\cD}, (\Pi + \Pi^{\perp}) A (\Pi + \Pi^{\perp})}\\
&= \Iprod{X_{\cD}, \Pi^{\perp} A \Pi^{\perp}}+ \Iprod{X_{\cD}, \Pi  A \Pi^{\perp}} + \Iprod{X_{\cD}, \Pi^{\perp} A \Pi} + \Iprod{X_{\cD}, \Pi A \Pi}.
\end{align*}
By assumption  $\cP$ is $2d$-complete on $\cD$ up to degree $2d$, and therefore $\Pi$ is derivable in degree $2d$ from the ideal constraints $\{g_j\}_{j \in [m]}$.
Therefore, the latter three terms may be absorbed into $G$, or more formally, we can set $A' = \Pi^{\perp} A \Pi^{\perp}$, $G' = G + (\Pi + \Pi^{\perp}) A (\Pi + \Pi^{\perp}) - \Pi^{\perp} A \Pi^{\perp}$, and re-write the original proof
\begin{align}
\sdpopt - \iprod{F, \hat x^{\tensor 2d}}&= \iprod{A',\hat x^{\tensor 2d}} + \iprod{G', \hat x^{\tensor 2d}} \label{eq:mId}.
\end{align}
    The left-hand-side remains unchanged, so we still have that it is bounded by $n^{O(d)}$ for any feasible solution $s \in \cS$.
Furthermore, the nonzero eigenspaces of $X_{\cD}$ and $A'$ are identical, and so $A'$ cannot be nonzero on any diagonal entry which is orthogonal to the space of feasible solutions.

    Now, we argue that every diagonal entry of $A'$ is at most $n^{O(d)}$.
    To see this, for each diagonal term $\chi_{\alpha}^2$, we choose the solution $s \in \cS$ for which $\chi_{\alpha}(s)^2 \ge n^{-O(d)}$.
    We then have by the PSDness of $A'$ that
    \[
	A'_{\alpha,\alpha} \cdot \chi_{\alpha}(s)^2 \le \iprod{s^{\tensor 2d}, A'} \le n^{O(d)},
    \]
which then implies that $A'_{\alpha,\alpha} \le n^{O(d)}$.
    It follows that $\Tr(A') \le n^{O(d)}$, and again since $A'$ is PSD,
\begin{align}
    \|A'\|_F \le \sqrt{\Tr(A')} \le n^{O(d)}.\label{eq:ftr}
\end{align}

Putting things together, we have from our original matrix identity \pref{eq:mId} that
\begin{align*}
\|G'\|_F
&= \|\sdpopt - A' - F \|_F\\
&\le \|\sdpopt\|_F + \|A'\|_F + \|F\|_F \quad (\text{triangle inequality})\\
    &\le \|\sdpopt\|_F + n^{O(d)} + \|F\|_F \quad (\text{from } \pref{eq:ftr}).
\end{align*}
Therefore by our assumptions that $\|\sdpopt\|,\|F\|_F = n^{O(d)}$, the conclusion follows.
\end{proof}

We now argue that the conditions of this theorem are met by several general families of problems.

\begin{corollary}\label{cor:well-cond}
    The following problems have degree-$2d$ \sos proofs with all coefficients bounded by $n^{O(d)}$:
    \begin{enumerate}
	\item The hypercube: Any polynomial optimization problem with the only constraints being $\{x_i^2 = x_i\}_{i\in[n]}$ or $\{x_i^2 = 1\}_{i\in[n]}$ and objective value at most $n^{O(d)}$ over the set of integer feasible solutions. (Including \maxkcsp).
	\item The hypercube with balancedness constraints: Any polynomial optimization problem with the only constraints being $\{x_i^2 - 1\}_{i\in[n]} \cup \{\sum_i x_i = 0\}$. (Including \communitydetection).
	\item The unit sphere: Any polynomial optimization problem with the only constraints being $\{\sum_{i\in[n]} x_i^2 = 1\}$ and objective value at most $n^{O(d)}$ over the set of integer feasible solutions. (Including \tensorpca).
	\item The sparse hypercube: As long as $2d \le k$, any polynomial optimization problem with the only constraints being $\{x_i^2 = x_i\}_{i\in[n]}\cup \{\sum_{i\in[n]} x_i = k\}$, or $\{x_i^3 = x_i\}_{i\in[n]}\cup \{\sum_{i\in[n]} x_i^2 = k\}$, and objective value at most $n^{O(d)}$ over the set of integer feasible solutions. (Including \densestksubgraph and the Boolean version of \sparsepca).
	\item The \maxclique problem.
    \end{enumerate}
\end{corollary}
We prove this corollary below.
For each of the above problems, it is clear that the objective value is bounded and the objective function has no large coefficients.
To prove this corollary, we need to verify the completeness of the constraint sets, and then demonstrate a set of feasible solutions so that each square term receives non-negligible mass from some solution.

A large family of completeness conditions were already verified by \cite{DBLP:journals/corr/RaghavendraW17} and others (see the references therein):

\begin{proposition}[Completeness of canonical polynomial optimization problems (from Corollary 3.5 of \cite{DBLP:journals/corr/RaghavendraW17})]\label{prop:comp1}
The following pairs of polynomial optimization problems $\cP$ and distributions over solutions $\cD$ are complete:
    \begin{enumerate}
	\item If the feasible set is $x \in \R^n$ with $\{x_i^2 = 1\}_{i\in[n]}$ or $\{x_i^2 = x_i\}_{i\in[n]}$, $\cP$ is $d$-complete up to degree $d$ (e.g. if $\cP$ is a CSP).
	    This is still true of the constraints $\{x_i^2 =1 \}_{i\in[n]}\cup\{\sum_i x_i = 0\}$ (e.g. if $\cP$ is a community detection problem).
	\item If the feasible set is $x \in \R^n$ with $\sum_{i\in[n]} x_i^2 = \alpha$, then $\cP$ is $d$-complete on $\cD$ up to degree $d$ (e.g. if $\cP$ is the tensor PCA problem).
	\item If $\cP$ is the \maxclique problem with feasible set $x \in \R^n$ with $\{x_i^2 = x_i\}_{i\in[n]}\cup\{x_ix_j = 0\}_{(i,j) \in E}$, then $\cP$ is $d$-complete on $\cD$ up to degree $d$.
    \end{enumerate}
\end{proposition}

A couple of additional examples can be found in the upcoming thesis of Benjamin Weitz \cite{BWthesis}:
\begin{proposition}[Completeness of additional polynomial optimization problems) \cite{BWthesis}]\label{prop:comp2}
The following pairs of polynomial optimization problems $\cP$ and distributions over solutions $\cD$ are complete:
    \begin{enumerate}
	\item If $\cP$ is the \densestksubgraph relaxation, with feasible set $x \in \R^n$ with $\{x_i^2 = x_i\}_{i \in[n]} \cup\{ \sum_{i\in[n]} x_i  = k\}$, $\cP$ is $d$-complete on $\cD$ up to degree $d \le k$.\label{cond:dks-complete}
	\item If $\cP$ is the \sparsepca relaxation with sparsity $k$, with feasible set $x \in \R^n$ with $\{x_i^3 = x_i\}_{i \in [n]}\cup \{\sum_{i\in[n]} x_i^2 = k\}$, $\cP$ is $d$-complete up to degree $d \le k/2$.\label{cond:spca-complete}
    \end{enumerate}
\end{proposition}

\begin{proof}[Proof of \pref{cor:well-cond}]
    We verify the conditions of \pref{thm:idealthing} separately for each case.
    \begin{enumerate}
	\item The hypercube: the completeness conditions are satisfied by \pref{prop:comp1}.
	    We choose the set of feasible solutions to contain a single point, $s = \vec{1}$, for which $\chi_{\alpha}^2(s) = 1$ always.
	\item The hypercube with balancedness constraints: the completeness conditions are satisfied by \pref{prop:comp1}.
	    We choose the set of feasible solutions to contain a single point, $s$, some perfectly balanced vector, for which $\chi_{\alpha}^2(s) = 1$ always.
	\item The unit sphere: the completeness conditions are satisfied by \pref{prop:comp1}.
	    We choose the set of feasible solutions to contain a single point, $s = \frac{1}{\sqrt{n}}\cdot \vec{1}$, for which $\chi_{\alpha}^2(s) \ge n^{-d}$ as long as $|\alpha| \le d$, which meets the conditions of \pref{thm:idealthing}.
	\item The sparse hypercube: the completeness conditions are satisfied by \pref{prop:comp2}.
	    Here, we choose the set of solutions $\cS = \{x \in \{0,1\}^n ~|~ \sum_i x_i = k\}$.
	    as long as $k > d$, for any $|\alpha| \le d$ we have that $\chi_S(x)^2 = 1$ when $s$ is $1$ on $\alpha$.
	\item The \maxclique problem: the completeness conditions are satisfied by \pref{prop:comp1}.
	    We choose the solution set $\cS$ to be the set of $0,1$ indicators for cliques in the graph.
	    Any $\alpha$ that corresponds to a non-clique in the graph has $\chi_\alpha$ identically zero in the solution space.
	    Otherwise, $\chi_{\alpha}(s)^2 = 1$ when $s \in \cS$ is the indicator vector for the clique on $\alpha$.
    \end{enumerate}
    This concludes the proof.
\end{proof}

\section{Lower bounds on the nonzero eigenvalues of some moment matrices}
In this appendix, we prove lower bounds on the magnitude of nonzero eigenvalues of covariance matrices for certain distributions over solutions.
Many of these bounds are well-known, but we re-state and re-prove them here for completeness.
We first define the property we want:
\begin{definition}
Let $\cP$ be a polynomial optimization problem and let $\cD$ be the uniform distribution over the set of feasible solutions $S$ for $\cP$.
Define the degree-$2d$ moment matrix of $\cD$ to be $X_{\cD} = \E_{x\sim\cD}[\hat{x}^{\tensor 2d}]$, where $\hat{x} = [1 \ x]^{\top}$.
\begin{itemize}
\item We say that \emph{$\cD$ is $\delta$-spectrally rich up to degree $2d$} if every nonzero eigenvalue of $X_\cD$ is at least $\delta$.
\end{itemize}
\end{definition}
\begin{proposition}[Spectral richness of polynomial optimization problems]\label{prop:spec-rich-examples}
The following distributions over solutions $\cD$ are polynomially spectrally rich:
    \begin{enumerate}
	\item If $\cD$ is the uniform distribution over $\{\pm 1\}^n$, then $\cD$ is polynomially spectrally rich up to degree $d \le n$.
	\item If $\cD$ is the uniform distribution over $\alpha \cdot \cS_{n-1}$, then $\cD$ is polynomially spectrally rich up to degree $d\le n$.
	\item If $\cD$ is the uniform distribution over $x \in \{1,0\}^n$ with $\|x\|_0 = k$, then if $2d \le k$, $\cD$ is polynomially spectrally rich up to degree $d$.
	\item If $\cD$ is the uniform distribution over $x \in \{\pm 1,0\}^n$ with $\|x\|_0 = k$, then if $2d \le k$,  $\cD$ is polynomially spectrally rich up to degree $d$.
    \end{enumerate}
\end{proposition}

\begin{proof}
    In the proof of each statement, denote the $2d$th moment matrix of $\cD$ by $X_\cD \defeq \E_{x\sim \cD}[x^{\tensor 2d}]$.
Because $X_{\cD}$ is a sum of rank-1 outer-products, an eigenvector of $X_{\cD}$ has eigenvalue $0$ if and only if it is orthogonal to every solution in the support of $\cD$, and therefore the zero eigenvectors correspond exactly to the degree at most $d$ constraints that can be derived from the ideal constraints.

    Now, let $p_1(x),\ldots,p_{r}(x)$ be a basis for polynomials of degree at most $2d$ in $x$ which is orthonormal with respect to $\cD$, so that
    \[
	\E_{x\sim\cD}[p_i(x) p_{j}(x)] = \begin{cases} 1 & i = j\\
	0 & \text{otherwise}\end{cases}
    \]
    If $\hat p_i$ is the representation of $p_i$ in the monomial basis, we have that
    \[
	(\hat p_i)^\top X_{\cD} \hat p_j
	= \E_{x \sim \cD} [p_i(x)p_j(x)].
    \]
Therefore, the matrix $R = \sum_{i} e_i (\hat p_i)^\top$ diagonalizes $X_{\cD}$,
    \[
	R X_{\cD} R^\top = \Id.
    \]
It follows that the minimum non-zero eigenvalue of $X_{\cD}$ is equal to the smallest eigenvalue of $(RR^\top)^{-1}$, which is in turn equal to $\frac{1}{\sigma_{\max}(R)^2}$ where $\sigma_{\max}(R)$ is the largest singular value of $R$.
Therefore, for each of these cases it suffices to bound the singular values of the change-of-basis matrix between the monomial basis and an orthogonal basis over $\cD$.
We now proceed to handle each case separately.
    \begin{enumerate}
	\item $\cD$ uniform over hypercube: In this case, the monomial basis \emph{is} an orthogonal basis, so $R$ is the identity on the space orthogonal to the ideal constraints, and $\sigma_{\max}(R) = 1$, which completes the proof.
	\item $\cD$ uniform over sphere: Here, the canonical orthonormal basis the spherical harmonic polynomials.
	    Examining an explicit characterization of the spherical harmonic polynomials (given for example in \cite{harmonics}, Theorem 5.1), we have that when expressing $p_i$ in the monomial basis, no coefficient of a monomial (and thus no entry of $\hat p_i$) exceeds $n^{O(d)}$\Tnote{}, and since there are at most $n^d$ polynomials each with $\sum_{i=0}^{d} \binom{n}{d} \le n^d$ coefficients, employing the triangle inequality we have that $\sigma_{\max}(R) \le n^{O(d)}$, which completes the proof.
	\item $\cD$ uniform over $\{x \in \{0,1\}^k ~|~ \|x\|_0 = k\}$:
	    In this case, the canonical orthonormal basis is the correctly normalized Young's basis (see e.g. \cite{DBLP:journals/combinatorics/Filmus16} Theorems 3.1,3.2 and 5.1), and agan we have that when expressing an orthonormal basis polynomial $p_i$ in the monomial basis, no coefficient exceeds $n^{O(d)}$.
	    As in the above case, this implies that $\sigma_{\max}(R) \le n^{O(d)}$ and completes the proof.
	\item $\cD$ uniform over $\{x \in \{\pm 1,0\}^k ~|~ \|x\|_0 = k\}$:
	    Again the canonical orthonormal basis is Young's basis with a correct normalization.
	    We again apply \cite{DBLP:journals/combinatorics/Filmus16} Theorems 3.1,3.2, but this time we calculate the normalization by hand: we have that in expressing each $p_i$, no element of the monomial basis has coefficient larger than $n^{O(d)}$ multiplied by the quantity
	    \[
		\E_{x\sim \cD}\left[\prod_{i=1}^d (x_{2i-1} - x_{2i})^2\right] = O(1).
	    \]
	This gives the desired conclusion.
    \end{enumerate}
\end{proof}
\section{From Boolean to Gaussian lower bounds}
\label{sec:reduce-to-hypercube}
In this section we show how to prove our \sos lower bounds for Gaussian PCA problems using the lower bounds for Boolean problems in a black-box fashion.
The techniques are standard and more broadly applicable than the exposition here but we prove only what we need.

The following proposition captures what is needed for tensor PCA; the argument for sparse PCA is entirely analogous so we leave it to the reader.

\begin{proposition}
  Let $k \in \N$ and let $A \sim \{\pm 1\}^{\binom{n}{k}}$ be a symmetric random Boolean tensor.
  Suppose that for every $A \in \{ \pm 1\}^{\binom nk}$ there is a degree-$d$ pseudodistribution $\pE$ satisfying $\{ \|x\|^2 = 1\}$ such that
  \[
  \E_A \pE \iprod{x^{\tensor k},A} = C\mper
  \]

  Let $T \sim \cN(0,1)^{\binom nk}$ be a Gaussian random tensor.
  Then
  \[
  \E_T \max_{\pE} \pE \iprod{x^{\tensor k}, T} \geq \Omega(C)
  \]
  where the maximization is over pseudodistributions of degree $d$ which satisfy $\{ \|x\|^2 = 1 \}$.
\end{proposition}
\begin{proof}
  For a tensor $T \in (\R^{n})^{\tensor k}$, let $A(T)$ have entries $A(T)_{\alpha} = \sign(T_\alpha)$.
  Now consider
  \[
  \E_T \pE_{A(T)} \iprod{x^{\tensor k}, T} = \sum_{\alpha} \E_T \pE_{A(T)} x^\alpha T_\alpha
  \]
  where $\alpha$ ranges over multi-indices of size $k$ over $[n]$.
  We rearrange each term above to
  \[
  \E_{A(T)} (\pE_{A(T)} x^\alpha) \cdot \E_{T_{\alpha} \, | \, A(T)} T_\alpha
  = \E_{A(T)} (\pE_{A(T)} x^\alpha) \cdot A(T)_\alpha \cdot \E |g|
  \]
  where $g \sim \cN(0,1)$.
  Since $\E |g|$ is a constant independent of $n$, all of this is
  \[
  \Omega(1) \cdot \sum_{\alpha} \E_A \pE_A x^\alpha \cdot A_\alpha = C\mper\qedhere
  \]
\end{proof}

\end{document}

